\documentclass[11pt,reqno]{amsart}

\usepackage{amsmath,amsfonts,amsbsy,amsthm}
\usepackage{amssymb}
\usepackage{mathabx}
\usepackage{eucal}
\usepackage{hyperref}
\usepackage{mathrsfs}
\usepackage{comment}
\usepackage{dsfont}
\usepackage{verbatim} 

\usepackage{enumerate}
\usepackage{enumitem}
\setlist{leftmargin=*}
\usepackage{tikz}
\usetikzlibrary{decorations.pathmorphing,shapes,arrows}

\numberwithin{equation}{section}

\makeatletter
\newtheoremstyle{corsivo}
   {\medskipamount}{\medskipamount}%
   {\itshape}{}%
   {\bfseries}{}%
   { }
   {\thmname{#1}\thmnumber{\@ifnotempty{#1}{ }\@upn{#2}}%
    \thmnote{ {\bfseries\boldmath(#3)}}.}%
\makeatother

\theoremstyle{corsivo}
\newtheorem{theorem}{Theorem}[section]
\newtheorem{lemma}[theorem]{Lemma}
\newtheorem{corollary}[theorem]{Corollary}
\newtheorem{proposition}[theorem]{Proposition}

\newtheorem{assumption}[theorem]{Assumption}

\makeatletter
\newtheoremstyle{dritto}
   {\medskipamount}{\medskipamount}%
   {\rmfamily}{}%
   {\bfseries}{}%
   { }
   {\thmname{#1}\thmnumber{\@ifnotempty{#1}{ }\@upn{#2}}%
    \thmnote{ {\bfseries\boldmath(#3)}}.}%
\makeatother

\theoremstyle{dritto}

\newtheorem{remark}[theorem]{Remark}

\newcommand{\sub}[1]{_{\mathrm{#1}}}

\newcommand{\su}[1]{^{\mathrm{#1}}}

\newcommand{\eps}{\varepsilon}
\newcommand{\epsi}{\varepsilon}

\newcommand{\Id}{\mathds{1}}   
\newcommand{\ex}{\mathrm{e}}
\newcommand{\iu}{\mathrm{i}}
\newcommand{\di}{\mathrm{d}}

\newcommand{\N}{\mathbb{N}}
\newcommand{\Z}{\mathbb{Z}}

\newcommand{\R}{\mathbb{R}}
\newcommand{\C}{\mathbb{C}}

\newcommand{\Hi}{\mathcal{H}}

\newcommand{\F}{\mathcal{F}}

\DeclareMathOperator*{\slim}{\text{s-}lim}

\newcommand{\scal}[2]{\left\langle\left.  #1 \right|#2 \right\rangle}                
\newcommand{\bscal}[2]{\left\langle  #1 \Big|#2 \right\rangle}        
     
\newcommand{\norm}[1]{\left\| #1 \right\|}

\newcommand{\bra}[1]{\left\langle #1 \right|}
\newcommand{\ket}[1]{\left| #1 \right\rangle}

\newcommand{\set}[1]{ \left\{  #1 \right\}} 

\DeclareMathOperator{\Tr}{Tr}         

 \DeclareMathOperator{\im}{Im}

\DeclareMathOperator{\ran}{Ran}

\newcommand{\ie}{{\sl i.\,e.\ }}   

\newcommand{\LH}{{\mathcal{L}(\Hi)}}
\newcommand{\Or}{{\mathcal{O}}}

\newcommand{\abs}[1]{\left\lvert#1\right\rvert}

\newcommand{\virg}[1]{``#1''}

\renewcommand{\(}{\left(}
\renewcommand{\)}{\right)}

\newcommand{\gm}{}
\newcommand{\hc}{}

\let\oldfootnote\footnote
\renewcommand{\footnote}[1]{\oldfootnote{\  #1}}

\title[A non-linear  Landauer--B\"uttiker formalism]{On the  self-consistent Landauer--B\"uttiker formalism}
\author[H. D. Cornean, G. Marcelli]{Horia D. Cornean \and Giovanna Marcelli}

\begin{document}

\begin{abstract}  We provide sufficient conditions such that the time evolution of a mesoscopic tight-binding open system with a local Hartree--Fock non-linearity converges to a self-consistent non-equilibrium steady state, which is independent of the initial condition from the \virg{small sample}. We also show that the steady charge current intensities are given by Landauer--B\"uttiker-like formulas, and make the connection with  the case of weakly self-interacting many-body systems.   
\end{abstract}

\maketitle

\goodbreak

\section{Introduction and the main results}
\subsection{The problem and its history}   The Landauer--B\"uttiker formalism \cite{B1, B2, IL, L1, L2} is one of the standard tools in the study of mesoscopic quantum transport. Boiled down to its essence, this formalism states that the steady charge current intensities  through different reservoirs (leads) -indirectly coupled through a \virg{small sample}- are  functions of certain quantum scattering transmission coefficients. These non-interacting steady current formulas  have been the object of a thorough mathematical investigation during the last two decades, both in the tight-binding setting \hc{\cite{HA}}, \cite{AJPP1, AJPP2, CJM, CGZ, CJN, CM, N}, and in the continuous one \cite{BP, CDNP, CDP, CNZ}. Quite recently, the authors of \cite{AJR} considered a very much related model where the time is also discretized. Note the important detail that while the physics community usually takes the existence of steady states for granted, providing a mathematical proof of this fact is not trivial, especially in the continuous setting. 

The situation in which the carriers are allowed to interact in the small sample is much more involved. The condensed-matter community widely uses the so-called non-equilibrium Green function (NEGF) formalism \cite{JWM, KKS, S}, which was only recently put on firm mathematical grounds \cite{CMP3}. In the locally interacting case, formulating a proper mathematical theory \hc{for the existence of non-equilibrium steady states (NESS) requires advanced mathematical techniques even for tight-binding models. One of the first papers which gives a list of generic sufficient conditions for the existence of NESS is \cite{R}. In this context, existence and completeness of M{\o}ller morphisms are much more difficult to prove than in the non-interacting case. They demand both a good control on the propagation estimates for the one particle Hamiltonian, and a clever way of dealing with an apparently exploding combinatorics of Dyson series \cite{BoMa, MB}. These ideas generated further activity on applications like linear response theory, correlation functions, and current formulas \cite{JP1, JP2, FMU, JPO, CMP2}}. 

\hc{Now let us go back to our concrete problem of obtaining Landauer-B\"uttiker type current formulas. To the best of our knowledge, for locally }interacting many-body mesoscopic systems this question has been for the first time analyzed in \cite{CMP1, CMP2}. The proofs are based on two main technical assumptions: first, the one-particle coupled Hamiltonian must have good enough dispersive estimates (and no bound-states), and second, the self-interaction must be small enough. A very interesting consequence of the results in \cite{CMP2} is that the steady current intensities are given, up to a second order error in the strength of the self-interaction, by a non-interacting Landauer--B\"uttiker formula, where the Hamiltonian of the small sample is perturbed by an additional potential which coincides with the first order iteration of a Hartree--Fock scheme. 
We note that self-consistent Hartree--Fock diagrammatic expansions are used by physicists to approximate Green--Keldysh functions with one-body objects (see \cite{JWM, SvL} and references therein). 

\subsection{What is new in this paper?}
\hc{We work in the so-called partitioning approach, see Figure \ref{apr13} for the generic setup. This means that when $t<0$, the total (decoupled) system consists of two isolated leads (actually any finite number of leads may be allowed) characterized by quasi-free equilibrium states, and a \virg{small} finite dimensional sample characterized by a self-consistent equilibrium state. At time $t=0$ the leads are coupled to the small sample, and for $t>0$ the total state is given by a density matrix $\rho(t)$ which solves a self-consistent Liouville equation. We are firstly interested in finding out whether  $\rho(t)$ has a limit when $t\to\infty$; to the best of our knowledge,  our paper is the first one providing rigorous results regarding this non-linear problem in the partitioning approach. The existence of a self-consistent steady state  is proved in Theorem \ref{thm:main}.  The convergence requires the same two main technical assumptions as in the weakly self-interacting many-body model \cite{CMP1, CMP2} discussed above. Also, this limit state is independent of the initial condition from the small sample. }

\hc{
A complementary approach to the partitioned self-consistent problem is the so-called partition-free \cite{Ni}, in which the leads and the sample are already coupled at $t<0$, and the non-linearity is added at $t=0$. A more detailed comparison between these two approaches can be found in Remark \ref{remark1.6}. 
}

\hc{The most important new practical application which we obtain in this paper, as its title suggests, may be found in Corollary \ref{coro2}, where we show that the steady state charge current intensity can still be expressed with a Landauer--B\"uttiker-like formula}. In Corollary \ref{coroHC} we derive an effective \virg{non-interacting}   formula (see also Corollary  \ref{Rhc1}), which replicates the weakly interacting many-body results.   We also comment on a self-consistent algorithm for computing the conductance \virg{near equilibrium} proposed in \cite{MN}; we explain in Remark \ref{Rhc2} how that algorithm can be justified within our framework.  

In the rest of this section we introduce the mathematical setting and formulate our main result, Theorem \ref{thm:main}, followed by a number of comments, open problems and corollaries. Section \ref{sec2} is entirely dedicated to the proof of Theorem \ref{thm:main}, which is based on a rather subtle fixed point argument. 
In Section \ref{sec3} we prove the various  Landauer--B\"uttiker formulas. 

 Technical results regarding the  global existence and uniqueness for the non-linear propagators are presented in Appendix \ref{Ap1}. In Appendix \ref{Ap2} we discuss which conditions are needed for the crucial dispersive bounds of Assumption \ref{as:main} to hold true. Some of these scattering estimates have been previously spelled out in \cite{CMP2}, but in a rather laconic manner. For completeness, we decided to give more details here. Finally, in Appendix \ref{Ap3} we analyse the continuity properties of the current density as a function of the energy.

\subsection{The configuration space}
By adopting the tight-binding approximation, we consider a discrete model given by two semi-infinite leads coupled to a finite system. In the following, the finite system will be named \emph{(small) sample}, while the term \emph{system} refers to the whole structure consisting of the sample together with the leads. The one-particle Hilbert space of the system is defined as 
\[
\Hi=\ell^2(\N_1) \oplus \ell^2(\N_2)\oplus \C^N,
\]
where $\N_j=\{0,1,2,\dots\}$ for $j\in\{1,2\}$,  and $N$ is a natural number.
The standard orthonormal basis of the $j$-th lead is denoted by $\{ \ket{n_j}:\, n\in\N_j\}$. The standard orthonormal basis of the sample $\C^N$ is denoted by $\{   \ket{\zeta_k}:\, k\in \{1,\dots, N\} \}$. \gm{As considered in the linear case \cite{CJM}, we emphasize that actually we can include finitely many leads, yielding analogous results.}

On each lead the dynamics is determined by the one-dimensional discrete Laplacian operator with Dirichlet boundary condition, which is denoted by $\Delta\su{D}$ and its definition is recalled below.
Let $t_c>0$ be the hopping constant, for every $\psi\in \ell^2(\N )$ one has that
\begin{equation} 
\label{eqn:lap}
\begin{aligned}
\(\Delta\su{D}\psi\)(n)&:=t_c\(\psi(n+1)+\psi(n-1)\)\text{ for all $n\geq 1$}\\
\(\Delta\su{D}\psi\)(0)&:=t_c\,\psi(1).
\end{aligned}
\end{equation} 
The operator $\Delta\su{D}$ is called the \emph{Dirichlet Laplacian}. We denote by $h_1$ and $h_2$ the Dirichlet Laplace operators acting on the first and on the second lead respectively.
\begin{figure}
\centering
\includegraphics[scale=0.19]{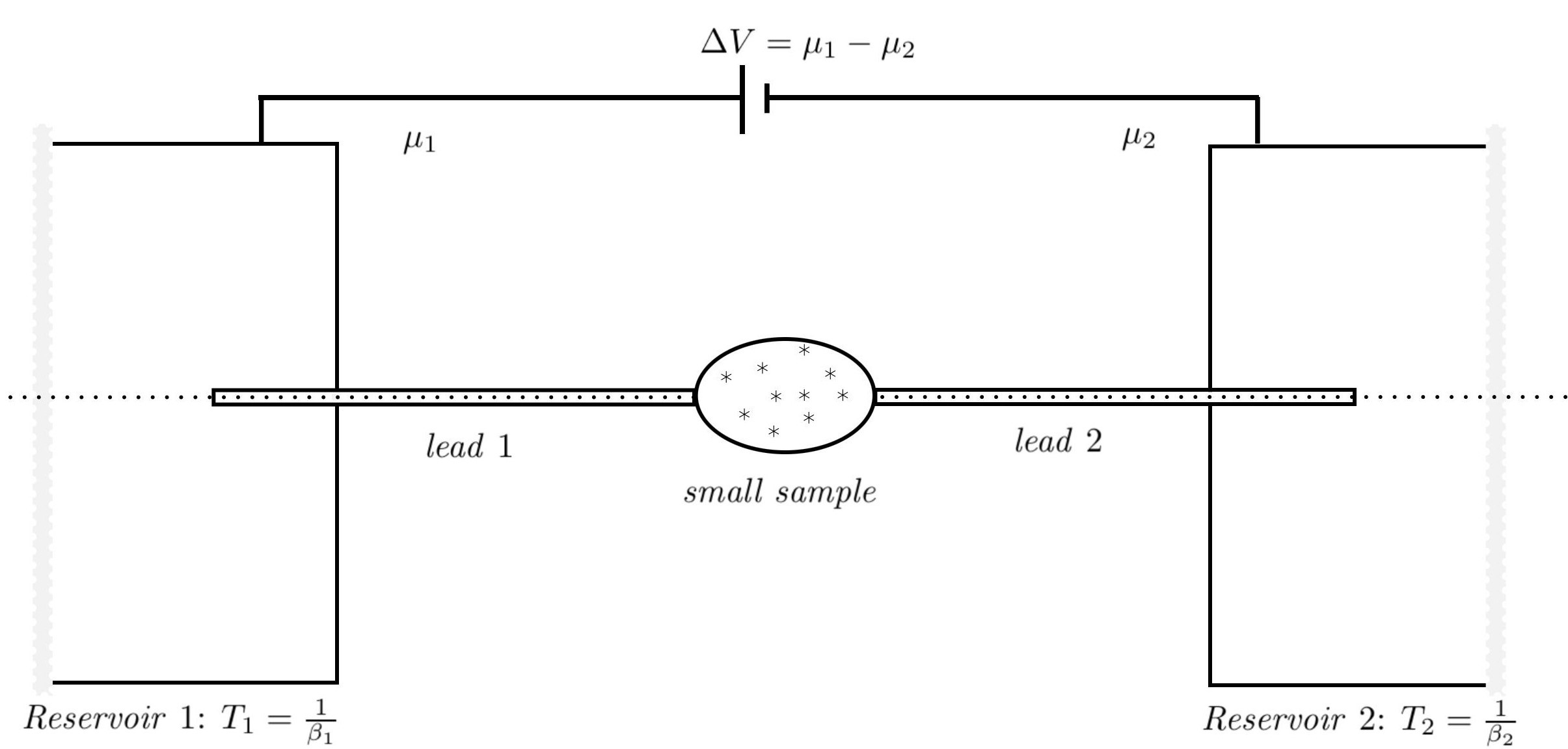} 
\caption{The generic setup.}
\label{apr13}
\end{figure}

\subsection{The initial equilibrium state of the sample}

We denote by $\mathcal{L}\(\C^N\)$ the set of the linear operators acting on $\C^N$. In the \virg{linear} case, the dynamics is given by some self-adjoint operator $h_s\in \mathcal{L}\(\C^N\)$. We assume that the sample is in contact with an energy reservoir fixing its temperature at value $1/\beta_s$ with $\beta_s>0$ and its mean number of particles is $0<\mathcal{N}<N$. In the non-interacting/linear case the equilibrium one-particle density matrix is
\[
\rho\su{non-int}_s:=\frac{1}{\ex^{\beta_s(h_s-\mu_s)}  +1 },
\]
where $\mu_s$ is the unique real solution of the equation
\begin{equation}
\label{gm1}
\Tr_{\C^N}\(\frac{1}{\ex^{\beta_s(h_s-x)}  +1 }\)=\mathcal{N},\qquad x\in\R,
\end{equation}
since the quantity on the left-hand side is increasing with $x$ and its range equals $(0,N)$.

We will now construct a whole class of self-consistent equilibrium density matrices. Fix $\nu_{jk}\in\R$ with $1\leq j,k\leq N$. Given some density matrix $\gamma\geq 0$ acting on $\C^N$, a non-linearity of Hartree type is given by a self-adjoint \virg{potential}
\begin{equation}
\label{eqn:V}
V_\lambda\{\gamma\}:=\lambda \sum_{j=1}^N\(  \sum_{k=1}^N \nu_{jk} \scal{\zeta_k}{\gamma \,\zeta_k}\)\ket{\zeta_j}\bra{\zeta_j}, \qquad \lambda\geq 0.
\end{equation}
We may also allow an \virg{exchange} non-linearity of the type 
$$\lambda  \sum_{1\leq j\neq k\leq N} \eta_{jk} \scal{\zeta_j}{\gamma \,\zeta_k}\ket{\zeta_k}\bra{\zeta_j},\quad \eta_{jk}=\overline{\eta_{kj}}\in \C,$$
but in order to simplify notation and because no extra mathematical challenges appear, we choose to only work with the Hartree term. \hc{We note that this type of self-consistent one-body effective potentials naturally appear from quartic many-body self-interactions when one performs a \virg{Wick partial contraction}, see for example Section 3.5 in \cite{CMP2}.  }

The self-consistent  but still linear sample Hamiltonian will be $h_s+V_\lambda\{\gamma\}$. The average number $\mathcal{N}$ of particles in the sample is chosen to stay fixed, hence the chemical potential of the sample has to solve the following equation:  
\[
\Tr_{\C^N}\(\frac{1}{\ex^{\beta_s(h_s+V_\lambda\{\gamma\}-x)}  +1 }\)=\mathcal{N},\qquad x\in\R.
\]
For a fixed $\gamma$, reasoning as in \eqref{gm1}, a solution (denoted by $\mu_s(\gamma)$) exists and is unique. We will prove the following result in the next section: 

\begin{lemma}\label{lemmahc10}
For every $\beta_s>0$, $\lambda\geq 0$ and  $\mathcal{N}\in (0,N)$, there exists at least one density matrix $\rho_s$ such that 
$$\rho_s=\frac{1}{\ex^{\beta_s(h_s+V_\lambda\{\rho_s\}-\mu_s(\rho_s))}  +1 },\quad {\rm Tr}(\rho_s)=\mathcal{N}.$$
In particular, $\rho_s$  commutes with $h_s+V_\lambda\{\rho_s\}$.
\gm{Moreover, there exists $\lambda_*>0$ such that if $0\leq \lambda\leq\lambda_*$ then such density matrix $\rho_s$ is unique.}
\end{lemma}
Any such $\rho_s$, parameterized by $\beta_s$, $\lambda$ and $\mathcal{N}$, can be our initial state in the sample. We will though see, that the constructed steady state will not depend on $\rho_s$ at all.

\subsection{The dynamics of the system}
The initial state $\rho_i$ of the full system is partitioned and defined as
\begin{equation}
\label{eqn:rhoi}
\rho_i:=\frac{1}{\ex^{\beta_1(h_1-\mu_1)}  +1 }\oplus \frac{1}{\ex^{\beta_2(h_2-\mu_2)}  +1 }\oplus \rho_s,
\end{equation}
where  $0<\beta_j\leq \infty$ and $\mu_j\in\R$ are constants. If $\beta_j=\infty$, the corresponding Fermi--Dirac distribution is replaced by $\chi_{\mu_j}(h_j)$, where $\chi_{\mu}$ is the indicator function of the interval $(-\infty,\mu]$.
\gm{We denote by $\LH$ the set of all linear and bounded operators from $\Hi$ to itself.}
Hereinafter, the Hartree potential $V_\lambda\{\,\cdot\,\}$ is understood as a map acting on non-negative density operators in $\LH$, by extending the definition given in \eqref{eqn:V} in the following way; the scalar product $\scal{\zeta_k}{\cdot\zeta_k}$ in the sample Hilbert space is extended to the one in the system Hilbert space and the projection $\ket{\zeta_j}\bra{\zeta_j}$ is seen as a projection on the vector $\zeta_j$ in the full system as well.
 
The stationary dynamics of the decoupled system is given by (we write $+$ instead of $\oplus$ from now on)
\begin{equation}\label{eqn:HD}
H_{D,\lambda}:=\gm{h_1+h_2+h_s+ V_\lambda\{\rho_i\}}.
\end{equation}
The coupling between the leads and the sample is realized through  a finite-rank \emph{tunneling Hamiltonian} of the type
\begin{equation}\label{dc6}
h_\tau:=\tau \sum_{j=1}^2\(  \ket{S_j}\bra{L_j}+ \ket{L_j}\bra{S_j}  \),\qquad \tau>0,
\end{equation}
\hc{where each $\ket{L_j}$ is  compactly supported in the $j$-th lead, and $\ket{S_j}$ is supported in the sample.}  
The operator describing the dynamics in the semi-infinite leads is denoted by
\begin{equation}\label{dc7}
H_L:=h_1+h_2.
\end{equation}
At $t=0$ we couple the leads to the sample. Let 
\begin{equation}
\label{eqn:H}
\gm{H:=h_1+h_2+h_s+h_\tau}
\end{equation}
be the linear coupled one-particle Hamiltonian. The time-dependent density operator will be given by the solution of the  Cauchy problem associate with the following non-linear Liouville equation:
\begin{gather} 
\label{gm2}
\begin{cases}
\iu\, \frac{\di}{\di t}\rho(t)=\big[ H+ V_\lambda \{\rho(t)\} , \rho(t)      \big],\quad t>0 \\
\rho(0)=\rho_i.
\end{cases}
\end{gather}
Defining the corresponding generator $G\colon \LH\to \LH$ such that
\begin{equation}
\label{eqn:G}
G(A):= H A+ V_\lambda \{A\, \rho_i\, A^*\} A,\qquad\text{for every $A\in \LH$}, 
\end{equation}
the above Cauchy problem boils down to the following one: find the differentiable (in operator norm topology) family of unitary operators $U(t)$ such that 
\begin{equation} 
\label{eqn:U}
\begin{cases}
\iu \frac{\di}{\di t}U(t)=G(U(t)),\qquad t>0 \\
U(0)=\Id.
\end{cases}
\end{equation}
Observe that in view of $\frac{\di}{\di t} \(   U^*(t)\)={\(\frac{\di}{\di t}    U(t)\)}^*$, we can write an equivalent Cauchy problem for $U^*(t)$:
\begin{equation} 
\label{eqn:U*}
\begin{cases}
-\iu \frac{\di}{\di t}U^*(t)={\(G(U(t))\)}^*=U^*(t)H+U^*(t)V_\lambda \{U(t)\, \rho_i\, U^*(t)\} ,\qquad t>0 \\
U^*(0)=\Id.
\end{cases}
\end{equation}
We will show in Appendix \ref{Ap1} that both \eqref{eqn:U} and \eqref{eqn:U*} have global solutions which are inverse to each other, hence $U(t)$ is unitary and the (unique) solution to the non-linear Liouville equation \eqref{gm2} is $\rho(t)=U(t)\rho_i U^*(t)$. Note that $U(t)$ also depends on $\rho_i$ in a non-trivial way. 

\gm{
\begin{remark}
\label{rem:expvalues} 
Notice that in \eqref{eqn:rhoi} the lead components of the initial state $\rho_i$ are bounded operators but not trace class. Indeed, if $\frac{1}{\ex^{\beta_j(h_j-\mu_j)}  +1 }$ was trace class, then it would be a compact operator with discrete spectrum, which would mean that $\Delta\su{D}$ also has discrete eigenvalues, which is in contradiction with Lemma \ref{lem:resolventlaplacian}. This is why in order to have a well-posed problem (see \eqref{eqn:mainq} below) we have to consider trace class observables.
\end{remark}
}

Hereinafter we denote by $S_1(\Hi)$ the trace class operators equipped with the norm 
$
\norm{O}_1:=\Tr(\sqrt{O^*O})$. Now we can finally formulate the main question we would like to answer. Given any trace class, self-adjoint observable $O$, does the following \emph{ergodic limit} exist:
\begin{equation}
\label{eqn:mainq}
\lim_{T\to \infty}\frac{1}{T}\int_0^T\di t\, \Tr\big (\rho(t) \, O\big ). 
\end{equation}

{ We are not able to answer this question in full generality, and we need to add some further assumptions on the system. The most important one is as follows:
\vspace{0.2cm}

\begin{assumption}
\label{as:main}  For every compactly supported functions $f,g\in \Hi$ we have $$\int_\R \di t\, \abs{\scal{f}{\ex^{\iu t H}g}}<\infty,$$ where $H$ is given in  \eqref{eqn:H}.
\end{assumption}

\vspace{0.2cm}

This also implies that $H$ has purely absolutely continuous spectrum because we automatically have a limiting absorption principle if $f$ has compact support: 
$$\scal{f}{(H-x-\iu\, 0_+)^{-1}f}:=\lim_{\eps\to 0^+}\scal{f}{(H-x-\iu\, \eps)^{-1}f}=\iu \int_0^{\infty} \di t\, \scal{f}{\ex^{-\iu \, t\, H}f}\, \ex^{\iu t x},$$ 
for every $x\in \R$.

\begin{figure}
\centering
\includegraphics[scale=0.2]{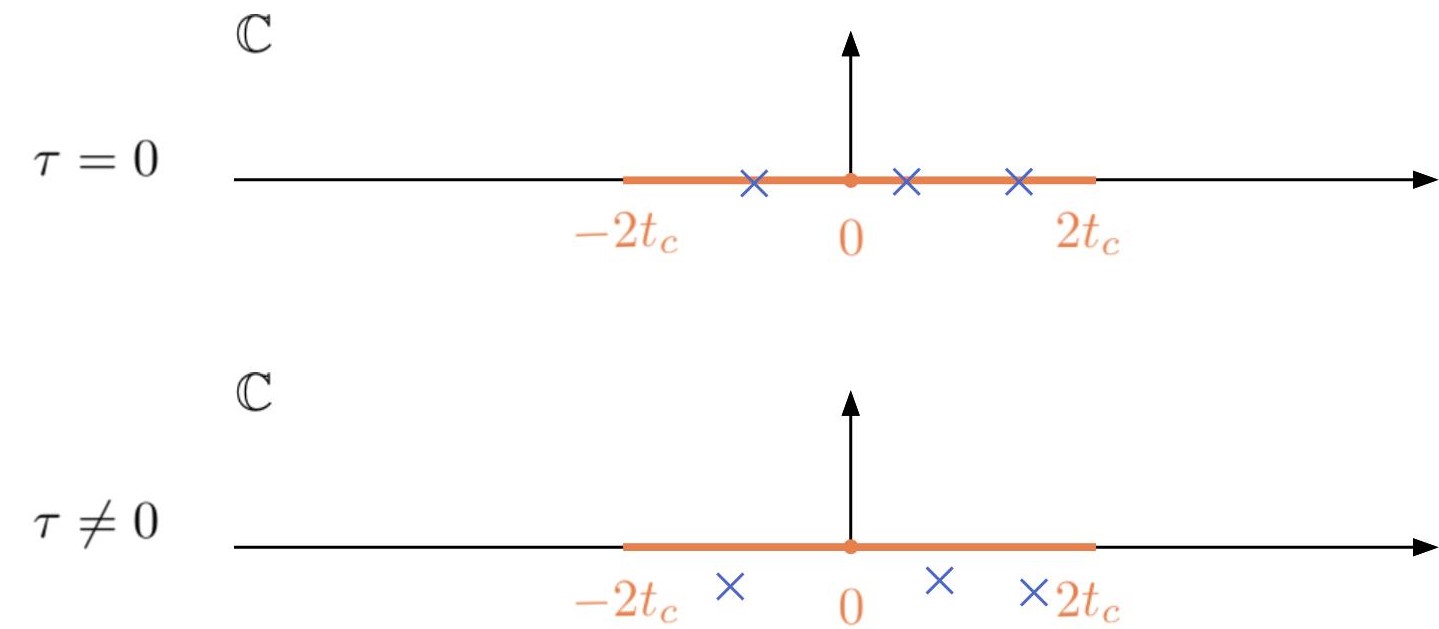}
\caption{The orange segment represents the absolutely continuous spectrum of the leads, while the blue crosses stand for the eigenvalues of $h_s$ when $\tau=0$, and the resonances of $H$ when $\tau\neq 0$.}
\label{apr14}
\end{figure}

The dispersive estimate of Assumption \ref{as:main} is not something one would normally encounter  in the physics literature, where the existence of a steady state is usually taken for granted.  Physicists are nevertheless aware that persistent oscillations may create problems, which are avoided when the coupled one-particle system is  \virg{fully resonant}, which means that the eigenvalues of $h_s$ are embedded in the spectrum of the leads and they become resonances when the tunneling $h_\tau$ is turned on, see Figure \ref{apr14}. In Appendix \ref{Ap2} we give some sufficient  conditions under which this rather strong assumption is satisfied. In particular, a detailed threshold analysis \cite{JK} is needed. The Dirichlet boundary conditions for the semi-infinite leads are important because they generate milder threshold singularities compared to the infinite leads, and the propagation estimates look as if we were in three dimensions for the free Laplacian.

In the following, we denote by $P\sub{ac}(H_L)$ the projection onto the absolutely continuous subspace associated with $H_L$, which is nothing but the projection on the two leads. This subspace coincides with the absolutely continuous subspace associated with the decoupled Hamiltonian $H_{D,\lambda}$  and is independent of $\lambda$. We also notice that $P\sub{ac}(H)$ (the projection onto the absolutely continuous subspace associated with $H$) equals the identity operator under our Assumption \ref{as:main}.

The next technical lemma is a direct consequence of \cite[Theorem 1 \& Corollary 2, \S 2, Ch. 6]{Y}, but it can be directly proved using the dispersive estimates from Appendix \ref{Ap2}: 
\begin{lemma}
\label{lem:waveop}
 The wave operators 
\begin{equation}
\label{eqn:waveop}
W_-(H,H_{L})=\slim_{t\to\infty} \ex^{-\iu t H} \ex^{\iu t H_{L}}P\sub{ac}(H_L),\quad W_-(H_{L},H)=\slim_{t\to\infty}\ex^{-\iu t H_{L}} \ex^{\iu t H}P\sub{ac}(H)
\end{equation}
exist, are complete and thus $W_-(H_{L},H)={\(W_-(H,H_{L})\)}^*$.
\end{lemma}

\vspace{0.2cm}

\subsection{The main results}
\begin{theorem}
\label{thm:main}
Let $\rho(t)\equiv \rho_\lambda(t)$ be the solution of \eqref{gm2}. Suppose that Assumption \ref{as:main} holds true. We define 
\begin{align}
\label{gm3}
M:=\max_{1\leq j,n\leq N}\int_0^\infty  \di s\, \abs{\scal{\zeta_j}{\ex^{\iu s H}\zeta_n}},\quad \norm{\nu}_1:=\sum_{j,k=1}^N \abs{\nu_{jk}},\quad 
\lambda_0:= \frac{1}{12 \norm{\nu}_1 M} 
.
\end{align}
Then for every $0\leq \lambda<\lambda_0$ there exists a family of steady states $\omega_\lambda: S_1(\Hi)\mapsto \C$  having a density operator $\rho_{\lambda,\infty} $ such that: 

\noindent {\rm (a)} $\omega_\lambda(O):=\lim_{t\to\infty}{\rm Tr}\big (\rho_\lambda (t)\, O\big )=: {\rm Tr}\big (\rho_{\lambda,\infty}\, O\big )$, see \eqref{dc20} for an explicit formula of $\rho_{\lambda,\infty}$.

\noindent {\rm (b)} $\rho_{\lambda,\infty}$ does not depend on the component $\rho_s$ from the initial state $\rho_i$.

\noindent {\rm (c)} The operator $H+V_\lambda\{\rho_{\lambda,\infty}\}$ commutes with $\rho_{\lambda,\infty}$.

\noindent {\rm (d)} Assume that $\beta_1=\beta_2=\beta$ and $\mu_1=\mu_2=\mu$. Let $\Pi_k$ denote the projection on lead $k$ and define the current intensity operator through lead $k$ as $I_k:=\iu [H,\Pi_k]$. 

Let $f_{\rm FD}(x)=1/(\ex^{\beta(x-\mu)}+1)$. Then 
\begin{equation}\label{dc25}
\rho_{\lambda,\infty}=f_{\rm FD}\big (H+V_\lambda\{\rho_{\lambda,\infty}\}\big )\quad \text{and}\quad \omega_\lambda(I_k)=0\, .
\end{equation}

\noindent {\rm (e)} $\rho_{0,\infty}=W_-(H,H_{L})\, \rho_i\, W_-(H_{L},H)$. 
\end{theorem}

\vspace{0.2cm}

\begin{remark} \label{remark1.6} A few comments are in place: 
\begin{itemize}

\item Due to the fact that $H$ does not have bound states and $\lambda$ is small, \gm{we obtain a pointwise convergence as $t\to\infty$ in (a) and thus it is not necessary to employ a C{\'e}saro limit like in \eqref{eqn:mainq}}. Such an  average is nevertheless needed in the non-interacting case when bound states are present, due to the persistent oscillations induced by the discrete eigenvalues  \cite{AJPP2, CGZ, CNZ, CNWZ}.

\item  The cases when $\lambda\neq 0$  and either one-particle bound states for $H$ are present, or the self-interaction is strong, remain two widely open problems. Some preliminary results for the self-interacting many-body problem were obtained in \cite{CM}, where the tunneling Hamiltonian $h_\tau$ was considered to be a perturbation to the decoupled self-interacting Hamiltonian. 

\item \hc{The existence of self-consistent steady-states was also investigated in \cite{Ni} for a continuum model in the partition-free approach. In that setting, the configuration space typically looks like a number of semi-infinite cylinders (which model the leads) connected through a bounded \virg{central} region. At $t<0$, the one-particle Hamiltonian $H_0$ is a Schr\"odinger like operator whose scalar potential might have different constant values inside the cylinders, which could model different chemical potentials. We emphasize the fact that no Dirichlet walls between the cylinders and the central region are present at any time.  Hypothesis 3.6 in \cite{Ni} demands the initial state $\rho_0$ not only to commute with $H_0$, but also, when expressed in the spectral representation of $H_0$, the initial state is assumed to be supported away from the possible eigenvalues of $H_0$ and from its scattering thresholds (as defined in  Mourre's commutator theory \cite{Mou}). At time $t=0$ a Hartree-type non-linearity is added. For $t>0$, the state $\rho(t)$ solves a self-consistent Liouville equation,  which under further assumptions, it is shown to admit a global solution in time. The existence of a self-consistent limit of $\rho(t)$ when $t\to\infty$ is not considered in  \cite{Ni}. Nevertheless, the existence of self-consistent steady states is shown in Theorem 6.4 using a Leray-Schauder fixed-point argument. Very roughly speaking, that result is a much more sophisticated and technically demanding version of what we do when we show the existence of a self-consistent $\rho_s$ in Lemma \ref{lemmahc10} using Brouwer's fixed-point theorem.
\item There are at least three other conceptual differences between our approach and that of \cite{Ni}. The first one is that in our case, we perturb an already self-consistent partitioned equilibrium state by turning on a coupling between leads and the small sample. The second one is that if we would require our initial state $\rho_i$ to \virg{not see} the eigenstates and thresholds of our initial decoupled Hamiltonian, then $\rho_s$ must equal zero and no particles are allowed in the small sample at $t<0$. Thirdly, our main interest is to derive Landauer-B\"uttiker-like formulas for the current intensity. A very interesting question is whether one could also derive such formulas in the framework of \cite{Ni}, and the answer is probably yes.}
\end{itemize}

\end{remark}

Now let us make the connection with the Landauer--B\"uttiker formalism and write down some relatively explicit formulas for the current intensity, which involve the transmission coefficient between the leads.

The operator describing the current intensity through lead $1$ is given by (see \eqref{dc6}): 
    $$I_1=\iu [H,\Pi_1]=\iu [h_\tau,\Pi_1]=\iu \, \tau\, \big (\ket{S_1}\bra{L_1}-\ket{L_1}\bra{S_1}\big).$$ 
\begin{corollary}\label{coro2}We employ the same hypotheses as in Theorem \ref{thm:main}. Then there exists a real-valued  function $\mathcal{T}_\lambda\in L^1([-2t_c,2t_c])$ \hc{(written in \eqref{gc13}, while $t_c$ is the hopping constant introduced in  \eqref{eqn:lap})} such that the steady state current intensity through lead 1 equals:
   \begin{equation}\label{hd1}
   \omega_\lambda(I_1)=2\pi \int_{-2t_c}^{2t_c} \Big ( \frac{1}{\ex^{\beta_2 (E-\mu_2)}+1}-\frac{1}{\ex^{\beta_1 (E-\mu_1)}+1}\Big )\, \mathcal{T}_\lambda(E)\, \di E.
   \end{equation}
    Moreover, the density $\mathcal{T}_\lambda$ seen as an element of $L^1([-2t_c,2t_c])$ has the following properties: 
   \begin{itemize}
   \item It admits a convergent expansion in powers of $\lambda$. The coefficient $\mathcal{T}_0$ (corresponding to $\lambda=0$) is independent of $\beta$'s and $\mu$'s, and equals the transmittance scattering coefficient between the two leads \hc{(see \eqref{march2})}. The other coefficients may depend on $\beta$'s and $\mu$'s (see also Corollary \ref{Rhc1} below).

\item 
  It is continuous with respect to  $\beta$'s and $\mu$'s, including the case where one or both $\beta$'s equal infinity (i.e. the Fermi--Dirac distribution is replaced by the corresponding indicator function).  

\item  Under the stronger assumptions of Lemma \ref{lemmahc1}, there  exists $0<\lambda_1\leq \lambda_0$ such that if $0\leq \lambda\leq \lambda_1$, then the density $\mathcal{T}_\lambda(\cdot)$ is actually continuous as a function of $E$. 
  
\end{itemize}

\end{corollary}

\hc{The next two Corollaries give approximate formulas for the steady values of various observables, up to an error of order $\lambda^2$. They are very close in spirit with the results of Theorem 3.5 in \cite{CMP2} which considers the locally interacting many-body problem. }

We introduce  an effective one-particle  Hamiltonian given by 
\begin{equation}
\label{dc1}
\begin{aligned}
V_{{\rm eff},\lambda}&:=V_\lambda\{W_-(H,H_{L})\, \rho_i\, W_-(H_{L},H)\}, \quad \lambda\geq 0, \\
H_{{\rm eff},\lambda}&:=H+ V_{{\rm eff},\lambda}.
\end{aligned}
\end{equation}
 Note that the wave operators $W_-(H_{L},H_{{\rm eff},\lambda})$ and $W_-(H_{ {\rm eff},\lambda},H_{L})$ exist and are complete since the conditions required in Lemma \ref{lem:waveop} are satisfied.
\begin{corollary}\label{coroHC}
    We work again under the same hypotheses as in Theorem \ref{thm:main}. Let $\lambda_0$ be as in \eqref{gm3}. Suppose that there exists $0<\lambda_1\leq \lambda_0$ such that for all $0\leq \lambda<\lambda_1$ the operator $H_{{\rm eff},\lambda}$ has purely absolutely continuous spectrum.  Define 
    \begin{equation}\label{dc15}
    \rho_{{\rm eff},\lambda}:=W_-(H_{{\rm eff},\lambda},H_{L})\, \rho_i\, W_-(H_{L},H_{{\rm eff},\lambda}).
    \end{equation}
    Then for every $f,g\in \Hi$ with compact support, and denoting by $O_c=\ket{f}\bra{g}$, we have 
    \begin{equation}\label{hm3}
    \omega_\lambda(O_c)={\rm Tr}\big (\rho_{{\rm eff},\lambda}\, O_c\big ) +\mathcal{O}(\lambda^2).
    \end{equation}
\end{corollary}
\hc{
\begin{remark}\label{march10}
Let us give a sufficient condition for the existence of such a $\lambda_1$ in Corollary \ref{coroHC}.
Let us assume that the matrix family $S(E)$ from Lemma \ref{lemmahc1} is invertible for all $E$. Then by replacing $h_s$ with $h_s+V_{{\rm eff},\lambda}$ in the definition of $S(E)$, we would obtain a new matrix family which remains invertible if $\lambda$ is sufficiently small. Thus, the operator $H_{{\rm eff},\lambda}$ has purely absolutely continuous spectrum for small enough $\lambda$. 
\end{remark}
}

\vspace{0.2cm}

    \begin{corollary}\label{Rhc1} Under \gm{the hypotheses of Corollary \ref{coroHC}}, let $\mathcal{T}_{{\rm eff},\lambda}(E)$ denote \hc{the \virg{non-interacting} transmittance coefficient between the two leads, with $H$ replaced by $H_{{\rm eff},\lambda}$ (see \eqref{march4} and \eqref{march2})}.  Then the density $\mathcal{T}_\lambda$  in \eqref{hd1}, seen as an element of  $L^1([-2t_c,2t_c])$,  obeys:
    \begin{equation}\label{hm2}
     \mathcal{T}_\lambda(E)= \mathcal{T}_{{\rm eff},\lambda}(E)\,  +\, \mathcal{O}(\lambda^2).
    \end{equation}
\end{corollary}

\vspace{0.2cm}

\begin{remark}\label{Rhc2}  In the physics literature dedicated to mesoscopic quantum transport it is common to assume that the temperatures and  chemical potentials of the leads are equal, \ie $\beta_1=\beta_2=\beta$ and $\mu_1=\mu_2=\mu$. In this case, one is not interested in the current intensity itself, which actually equals zero  according to \eqref{dc25} or \eqref{hd1}. Instead, one would like to compute the \virg{near equilibrium} conductance, defined as the ratio between the current intensity and the potential drop $\mu_1-\mu_2$ in the limit when both $\mu_1$ and $\mu_2$ converge to $\mu$. At zero temperature ($\beta\to\infty$), and if $\lambda$ is small enough, this coincides with $2\pi\,\lim_{\beta\to\infty} \mathcal{T}_\lambda(\mu)$, as it can be inferred from Corollary \ref{coro2}. At zero temperature, in two-dimensional quantum Hall systems and Chern insulators the adiabatic charge
transport is usually investigated by the Kubo formula in terms of the Hall conductance
\cite{ES, GM1} or equivalently the Hall conductivity which remains valid beyond the linear
response regime \cite{GM2} for non-interacting fermionic systems.

  Let us recast the \virg{near equilibrium} approach presented in \cite{MN} within our framework. In spite of the fact that their method does not consider time evolution at all, its range of validity can be analyzed by using our mathematical results.

\noindent
 We start by introducing the map 
$$[0,1]^N\ni (n_1,...,n_N)=:\vec{n}\mapsto H_\lambda(\vec{n}):=H+\lambda\sum_{j=1}^N\Big (\sum_{k=1}^N \nu_{jk}n_k\Big )\, \ket{\zeta_j}\bra{\zeta_j}.$$
\hc{Note that the term added to $H$ is exactly a Hartree interaction like in \eqref{eqn:V}, where $\scal{\zeta_k}{\gamma\,  \zeta_k}$ is now replaced by the \virg{occupation numbers} $0\leq n_k\leq 1$.} The authors of \cite{MN} make the assumption that if a steady state is achieved, and because both leads have the same temperatures and chemical potentials, 
then one expects that the steady state \hc{density matrix $\rho_{\lambda,\infty}$} is given by the \virg{thermal equilibrium} self-consistent density operator $f_{\rm FD}(H_\lambda(\vec{n}_\lambda))$, where $f_{\rm FD}$ is the Fermi--Dirac distribution and  $\vec{n}_\lambda$ contains the steady values of the occupation numbers in the small sample, determined by the following fixed point condition: 
\begin{align}
\label{eqn:8}
\hc{\scal{\zeta_k}{\rho_{\lambda,\infty}\, \zeta_k}=}n_{\lambda, k} =\scal{\zeta_k}{f_{\rm FD}(H_{\lambda}(\vec{n}_{\lambda}))\, \zeta_k},\quad 1\leq k\leq N.
\end{align}
This equation is nothing but our  self-consistent relation \eqref{dc25}, when  $\rho_{\lambda,\infty}$ is restricted to the small sample. 
In view of functional calculus via the resolvent formalism, one can rewrite \eqref{eqn:8} as
\begin{equation}
\label{dc12}
 n_{\lambda, k} =\frac{1}{\pi} \int_{-2t_c}^{2t_c}\di E \frac{1}{\ex^{\beta (E-\mu)}+1}{\rm Im}\scal{\zeta_k} {\big (H_\lambda(\vec{n}_{\lambda, k})-E-\iu 0_+\big )^{-1}\zeta_k}.
\end{equation}
The matrix elements on the right-hand side of \eqref{dc12} can be computed using the Feshbach formula  (see \eqref{eqn:resformula-bis}), which reduces to finding the inverse of the following $N\times N$ matrix
$$
S(E)+\lambda\sum_{j=1}^N\Big (\sum_{k=1}^N \nu_{jk}n_{\lambda, k}\Big )\, \ket{\zeta_j}\bra{\zeta_j},
$$
\gm{with 
\[
S(E):=h_s -E - \tau^2 \sum_{j=1}^2  \bra{L_j}  (h_j-E-\iu 0_+)^{-1}   \ket{L_j}\, \ket{S_j}\bra{S_j}
\]
as defined in Lemma \ref{lemmahc1}}. Thus \eqref{dc12}  coincides with formula (2.11) in \cite{MN}, provided that the temperature is taken to zero ($\beta\to \infty$) and $\mu$ is the fixed Fermi energy of the leads. 

By iterating, the authors of \cite{MN} find a numerical solution $\vec{n}_\lambda$ to \eqref{dc12}. This solution is then directly plugged into the {\it non-interacting} Landauer--B\"uttiker conductance formula, where $H$ is replaced by $H_\lambda(\vec{n}_\lambda)$ (see \cite[Eq. (2.9)]{MN}).

Let us now explain how this apparently ad-hoc second step of their algorithm for computing the conductance can also be understood within our framework, through Corollary \ref{Rhc1}, when $\lambda$ is small enough. Using the intertwining property of wave operators in \eqref{dc15} we get 
\begin{equation}\label{dc10}
\rho_{{\rm eff},\lambda}=f_{\rm FD}(H_{{\rm eff},\lambda}).
\end{equation}
Denote by $s_k:=\scal{\zeta_k}{W_-(H,H_L)\rho_i W_-(H_L,H)\zeta_k}=\scal{\zeta_k}{f_{\rm FD}(H)\, \zeta_k}$. From \eqref{dc1} we obtain that $H_{{\rm eff},\lambda}= H_{\lambda}(\vec{s})$. From regular perturbation theory we get
\begin{equation*}
s_j+\mathcal{O}(\lambda)=\scal{\zeta_j}{f_{\rm FD}(H_{{\rm eff},\lambda})\, \zeta_j}.
\end{equation*}
Thus, from \eqref{dc10} we have 
\begin{equation}\label{dc13}
s_j+\mathcal{O}(\lambda)=\scal{\zeta_j}{f_{\rm FD}(H_{{\rm eff},\lambda})\, \zeta_j}=\scal{\zeta_j}{\rho_{{\rm eff},\lambda}\zeta_j}=\scal{\zeta_j}{f_{\rm FD}(H_{\lambda}(\vec{s}))\, \zeta_j},\, 1\leq j\leq N.
\end{equation}

One can show (we do not give details here) that if $\lambda$ is small, the right-hand side of \eqref{dc12} defines a contraction on $[0,1]^N$ and has a unique fixed point $\vec{n}_\lambda$. Also, since \eqref{dc13} shows that $\vec{s}$ is an \virg{almost} fixed point for \eqref{dc12} up to an error of order $\lambda$, then  $\vec{n}_\lambda=\vec{s} +\mathcal{O}(\lambda)$. 
This implies that the difference between $H_\lambda(\vec{s})$ and $H_\lambda(\vec{n}_\lambda)$ is of order $\lambda^2$. 

Now Corollary \ref{Rhc1} implies that the steady state value of the conductance can be computed (up to an error of order $\lambda^2$) by using $H_\lambda(\vec{s})$ instead of $H$ in the non-interacting Landauer--B\"uttiker formula. One gets the same conclusion by using $\vec{n}_\lambda$ instead of $\vec{s}$, because the difference between $H_\lambda(\vec{s})$ and $H_\lambda(\vec{n}_\lambda)$ is of order $\lambda^2$. 
 Thus if $\lambda$ is small enough, also the second step of the algorithm of \cite{MN} can be justified within our framework, up to errors of order $\lambda^2$.    

On the other hand, the  method of \cite{MN} does not work if either the chemical potentials or the temperatures on the leads are not equal. In this case, even though $\rho_{\lambda,\infty}$ and $H+V_\lambda\{\rho_{\lambda,\infty}\}$ still commute with each other according to Theorem \ref{thm:main}(c), $\rho_{\lambda,\infty}$ cannot be written as a function of $H+V_\lambda\{\rho_{\lambda,\infty}\}$. Thus the problem of finding $\rho_{\lambda,\infty}$ can no longer be reduced  to a fixed point equation like in \eqref{dc12}, where the only unknowns are the steady state occupation numbers of the small sample.  Nevertheless, our formulas \eqref{hm2} and \eqref{hm3} still hold true also in the general case. 
\end{remark}

\vspace{0.2cm} \noindent {\bf Acknowledgements}. Both authors acknowledge support from the Independent Research Fund Denmark--Natural Sciences,
grant DFF–10.46540/2032-00005B. 

\vspace{0.2cm} \noindent {\bf Data Availability Statement}. 
Data sharing not applicable to this article as no datasets were generated or
analysed during the current study.

\vspace{0.2cm} \noindent {\bf Declarations}.\\
\noindent {\bf Conflict of interest}.
The authors have no competing interests to declare that are relevant to the content of this article.

\section{Proof of Theorem \ref{thm:main}}\label{sec2}

\subsection{Proof of Lemma \ref{lemmahc10}} We define the set 
\[
K:=\set{ \gamma\in\mathcal{L}\( \C^N \): \gamma\geq 0,\quad \Tr(\gamma)=\mathcal{N}    }.
\]
Since every self-adjoint matrix in $\C^{N\times N}$ has $N(N+1)/2$ independent components and the diagonal entries are real, we conclude that the set $K$ can be identified with a subset of $ \R^{N^2}$. The largest eigenvalue (thus the norm) of each $\gamma$ is bounded by $\mathcal{N}$, hence the absolute values of all entries $\scal{\zeta_k}{\gamma\, \zeta_j}$ has the same property. This shows that $K$ is bounded. 

Let us show that $K$ is convex. Indeed, any convex combination of two $\C^{N\times N}$ non-negative (thus self-adjoint) bounded matrices with equal traces, will remain non-negative and have the same trace. 

Now let us show that $K$ is closed. First, the component-wise point convergence in $K$ seen as a subset of $\R^{N^2}$ is just the weak convergence for the corresponding density matrices, which is equivalent with the Hilbert--Schmidt norm convergence because $N<\infty$. Second, since the (real) spectrum of a given self-adjoint $\gamma$  varies continuously (in the Hausdorff distance) with $\gamma$, the infimum of the spectrum is continuous with respect to $\gamma$. Thus both the value of the trace and the non-negativity of the spectrum are preserved by taking limits in $K$. Hence $K$ is closed (thus also compact).

The map 
$$T:\R^{N^2}\times \R\mapsto \R^{N^2},\qquad T(\gamma,x):=\frac{1}{\ex^{\beta_s(h_s+V_\lambda\{\gamma\}-x)}  +1 }$$
is smooth; this can be seen by writing $$T(\gamma,x)=\frac{\iu}{2\pi}\int_{\mathcal{C}} \di z\,\frac{1}{\ex^{\beta_s(z-x)}+1}\big (h_s+V_\lambda\{\gamma\}-z\big)^{-1}$$
where the positively oriented simple contour $\mathcal{C}$ is included in the strip $|\text{Im}(z)|\leq \pi/(2\beta_s)$ and encircles the real eigenvalues of $h_s+V_\lambda\{\gamma\}$. Then by regular perturbation theory applied to the resolvent appearing in the integrand, all the matrix elements of $T(\gamma,x)$ are differentiable. By using the implicit function theorem, we see that the chemical potential $\mu_s(\gamma)$ which denotes the unique solution of ${\rm Tr}\big (T(\gamma,\mu_s(\gamma)\big )=\mathcal{N}$ must also be  continuous as a function of $\gamma$.  

Now let us consider the map
\begin{equation}\label{dc11}
K\ni \gamma\mapsto F(\gamma):=T(\gamma,\mu_s(\gamma))=\frac{1}{\ex^{\beta_s(h_s+V_\lambda\{\gamma\}-\mu_s(\gamma))}  +1 }\in K.
\end{equation}

The above map $F$ is continuous and leaves the convex and compact set $K$ invariant.   Brouwer's fixed point theorem implies that there exists at least one fixed point $\rho_s\in K$. 

\gm{Let us now show that if $\lambda$ is sufficiently small then the map $F$ is a contraction, thus $\rho_s$ is unique.
We will prove that there exists a constant $C_*$ such that
\begin{equation}
\label{eqn:Fcontraction}    
\norm{F(\gamma_1)-F(\gamma_2)}\leq C_* \lambda\norm{\gamma_1 -\gamma_2}\qquad\text{ for all $\gamma_1,\gamma_2\in K$.}
\end{equation}
To prove inequality \eqref{eqn:Fcontraction}, we notice that
\begin{equation}
\label{eqn:Fcontractionsplit}
\begin{aligned}
\norm{F(\gamma_1)-F(\gamma_2)}\leq& \norm{T(\gamma_1,\mu_s(\gamma_2))-T(\gamma_2,\mu_s(\gamma_2))}\\
&+\norm{T(\gamma_1,\mu_s(\gamma_1))-T(\gamma_1,\mu_s(\gamma_2))}.
\end{aligned}
\end{equation}
To estimate both summands, denoting by $P_j:=\ket{\zeta_j}\bra{\zeta_j}$ we preliminary observe that
\begin{equation}
\label{eqn:dergamma}
\partial_{\gamma_{ij}}T(\gamma,x)=-\delta_{ij}\frac{\iu\lambda\nu_{jj}}{2\pi}\int_{\mathcal{C}} \di z\frac{1}{\ex^{\beta_s(z-x)}+1}\big (h_s+V_\lambda\{\gamma\}-z\big)^{-1}P_j\big (h_s+V_\lambda\{\gamma\}-z\big)^{-1}
\end{equation}
and thus there exists a constant $C_1$ (which can be chosen independent of a sufficiently small $\lambda$) such that
\begin{equation}
\label{eqn:derT}  
\sum_{i,j=1}^N\norm{\partial_{\gamma_{ij}}T(\gamma,x)}\leq C_1\lambda\quad\text{for all $\gamma\in K,\, x\in \R$.}
\end{equation}
Therefore, for the first summand on the right-hand side of inequality \eqref{eqn:Fcontractionsplit} we get that 
\begin{align*}
\norm{T(\gamma_1,\mu_s(\gamma_2))-T(\gamma_2,\mu_s(\gamma_2))}\leq \sum_{i,j=1}^N\max_{\gamma \in K}\norm{\partial_{\gamma_{ij}}T(\gamma,\mu_s(\gamma_2))}\norm{\gamma_1-\gamma_2}\leq C_1\lambda\norm{\gamma_1-\gamma_2}.    
\end{align*}
For the second summand on the right-hand side of \eqref{eqn:Fcontractionsplit}, we observe that
\begin{align*}
\norm{T(\gamma_1,\mu_s(\gamma_1))-T(\gamma_1,\mu_s(\gamma_2))}\leq \max_{\gamma\in K, x\in\R}\norm{\partial_x  T(\gamma,x)}\abs{\mu_s(\gamma_1)-\mu_s(\gamma_2)}.
\end{align*}
Reasoning as before, applying the implicit function theorem and inequality \eqref{eqn:derT}, we can find a constant $C_2$ such that
\begin{align*}
\abs{\mu_s(\gamma_1)-\mu_s(\gamma_2)}
\leq C_2\lambda\norm{\gamma_1-\gamma_2},
\end{align*}
which once inserted in the last inequality, it ends the proof of \eqref{eqn:Fcontraction}.}
\qed

\subsection{Proof of Theorem \ref{thm:main}{(a),(b)}} In view of the density of finite-rank operators in $S_1(\Hi)$, it is enough to prove the limit in Theorem \ref{thm:main}{(a)} for the case in which $O$ equals a rank-one projection $\ket{f}\bra{g}$ for some unit vectors $f,g\in\Hi$ with compact support. We may write 
\begin{align}\label{dc3}
    {\rm Tr}\big (\rho(t)\ket{f}\bra{g}\big )&=\scal{U^*(t)g}{\rho_iU^*(t)f}\nonumber \\
    &=\scal{\ex^{-\iu tH}U^*(t)g}{\Big (\ex^{-\iu tH}\rho_i \ex^{\iu tH}\Big )\ex^{-\iu tH}U^*(t)f}.
\end{align}
Because $\rho_i$ and $H_{L}$ commute (see \eqref{eqn:rhoi} and \eqref{dc7}) we have 
$$\ex^{-\iu tH}\rho_i \ex^{\iu tH}=\ex^{-\iu tH}\ex^{\iu tH_{L}}\rho_i \ex^{-\iu tH_{L}}\ex^{\iu tH}.$$
 Using Lemma \ref{lem:waveop}, and that $P_{\rm ac}(H)$ is the identity operator, we have that $\ex^{-\iu tH_{L}}\ex^{\iu tH}$ converges strongly to $W_-(H_{L},H)$, an operator which maps onto the lead-space, hence we may insert a $P_{\rm ac}(H_{L})$ just after $\rho_i$. We also have (see \eqref{eqn:rhoi}): 
 \begin{equation}
 \label{dc2}
\rho_i\,  P_{\rm ac}(H_{L})=P_{\rm ac}(H_{L})\, \frac{1}{\ex^{\beta_1(h_1-\mu_1)}  +1 }\oplus \frac{1}{\ex^{\beta_1(h_2-\mu_2)}  +1 }\oplus {\bf 0}.
\end{equation}
Thus 
\begin{equation}\label{dc4}
\slim_{t\to\infty} \ex^{-\iu tH}\rho_i \ex^{\iu tH}=W_-(H,H_{L})\Big (\frac{1}{\ex^{\beta_1(h_1-\mu_1)}  +1 }\oplus \frac{1}{\ex^{\beta_2(h_2-\mu_2)}  +1 }\oplus \, {\bf 0}\Big )\, W_-(H_{L},H),
\end{equation}
which no longer depends on $\rho_s$.

In order to show that the right-hand side of \eqref{dc3} has a limit, we need to investigate the strong limit of the following operator:   
\begin{equation}
A_\lambda(t):=\ex^{-\iu t H}U^*(t),
\end{equation}
where the $\lambda$ dependence appears through $U(t)$. 


\begin{proposition}
\label{prop:W}
Let $H$ be as in \eqref{eqn:H} and $V_\lambda$ as in \eqref{eqn:V}. Suppose that Assumption \ref{as:main} holds true. Let $\lambda_0$ be defined as in \eqref{gm3}. Then for any $0\leq \lambda<\lambda_0$ we have that
\[
A_{\lambda,\infty}:=\slim_{t\to\infty} A_\lambda(t)\quad\text{exists.}
\]
\end{proposition}
\begin{proof}
Since $A_\lambda(t)$ is unitary, it is enough to prove the existence of a strong limit on compactly supported functions. By using the fundamental theorem of calculus and Cauchy problem \eqref{eqn:U*} for $U^*(t)$ (whose existence and uniqueness of solution is guaranteed by Proposition \ref{prop:exunU}), we get that
\begin{equation}\label{dc21}
U^*(t)\ex^{-\iu t H}=\Id +\iu \lambda\sum_{1\leq j,k\leq N}\nu_{jk}\int_0^t \di s \,\scal{\zeta_k}{U(s)\,\rho_i\, U^*(s) \zeta_k} U^*(s)\ket{\zeta_j}\bra{\zeta_j}\ex^{-\iu s H}.
\end{equation}
By multiplying the above equality with $\ex^{-\iu t H}$ on the left-hand side and with $\ex^{\iu t H}$ on the right-hand side, we have that
\begin{align}\label{hm1}
&A_{\lambda}(t)=\Id \\
&\quad \nonumber +\iu\lambda\sum_{1\leq j,k\leq N} \nu_{jk}\int_0^t \di s\, \scal{A_{\lambda}(s)\zeta_k}{\ex^{-\iu s H} \,\rho_i\, \ex^{\iu s H} A_{\lambda}(s) \zeta_k} \ex^{\iu (s-t
) H} A_{\lambda}(s)\ket{\zeta_j}\bra{\zeta_j}\ex^{\iu (t-s) H}.
\end{align}
The right-hand side of \eqref{hm1}  only involves the restriction of   $A_\lambda(t)$ to the subspace of the small sample. Showing first that this restriction has a limit when $t\to\infty$ is the main idea.   We thus  introduce the map
\[
[0,\infty)\ni t\mapsto a_\lambda(t):=\( A_{\lambda}(t)\ket{\zeta_1},\dots,A_{\lambda}(t)\ket{\zeta_N}\)\in\Hi^N.
\]
Applying the left-hand side of \eqref{hm1} on $\ket{\zeta_n}$  for all  $1\leq n\leq N$ leads to:
\begin{align*}
& a_{\lambda,n}(t)=\ket{\zeta_n}\\
&\quad +\iu \lambda\sum_{1\leq j,k\leq N}\nu_{jk}\int_0^t \di s\, \scal{a_{\lambda,k}(s)}{\ex^{-\iu s H} \,\rho_i\, \ex^{\iu s H} a_{\lambda,k}(s)}\scal{\zeta_j}{\ex^{\iu (t-s) H}\zeta_n} \ex^{\iu (s-t
) H} a_{\lambda,j}(s).
\end{align*}

Let us consider the space 
\[
C_b([0,\infty), \Hi^N):=\{ f\colon [0,\infty) \to \Hi^N,\,\text{ $f$ is continuous and bounded}  \}
\]
equipped with the norm
\begin{equation}
\label{4}    
\norm{f}_\infty:=\sup_{t\in [0,\infty) } \max_{1\leq n\leq N} \norm{f_n(t)}. 
\end{equation}
Since $\norm{a_\lambda}_\infty=1$, the function $a_\lambda$ is in $C_b([0,\infty), \Hi^N)$ and can been seen as a fixed point of the map $\Phi\colon C_b([0,\infty), \Hi^N) \to C_b([0,\infty), \Hi^N)$  defined as
\begin{equation}
\label{2}
\begin{aligned}
&{\(\Phi(f)\)}_n(t):=\\
&\quad\ket{\zeta_n}+\iu \lambda\sum_{1\leq j,k\leq N}\nu_{jk}\int_0^t \di s\, \scal{f_k(s)}{\ex^{-\iu s H} \,\rho_i\, \ex^{\iu s H} f_k(s)}\scal{\zeta_j}{\ex^{\iu (t-s) H}\zeta_n} \ex^{\iu (s-t) H} f_j(s).
\end{aligned}
\end{equation}
Denoting by $B_2(0)$ the closed ball of radius $2$ centered at $0$ in $\Hi^N$, consider the subspace $C([0,\infty), B_2(0))\subset  C_b([0,\infty), \Hi^N)$. We will show that there exists a positive $\lambda_0$ such that by choosing $\lambda<\lambda_0$ the map $\Phi$ leaves invariant $C([0,\infty), B_2(0))$ and is a contraction on $C([0,\infty), B_2(0))$. 
Indeed, for every $f\in C([0,\infty), B_2(0))$ we get that
\[
\norm{{\(\Phi(f)\)}_n(t)}\leq 1+\lambda \norm{\nu}_1 \norm{f}^3_\infty \int_0^t  \di s\, \abs{\scal{\zeta_j}{\ex^{\iu (t-s) H}\zeta_n}}\leq 1+8\lambda \norm{\nu}_1  M.
\] 
Similarly, one can easily obtain that for any $f,g\in C([0,\infty), B_2(0))$ it holds true that
\begin{align*}
\norm{{\(\Phi(f)\)}_n(t)-{\(\Phi(g)\)}_n(t)}&\leq 12\lambda \norm{\nu}_1 M  \norm{f-g}_\infty.
\end{align*}
Therefore, Banach--Caccioppoli fixed-point theorem implies that $a_\lambda$ is the unique fixed point of $\Phi\colon C([0,\infty), B_2(0)) \to C([0,\infty), B_2(0))$. Hence, $a_\lambda$ can be recovered by iteration and equals the limit of the iterated sequence:
\[
a_\lambda^{p+1}=\Phi(a_\lambda^{p}),\qquad a_\lambda^{0}=\( \ket{\zeta_1},\dots,\ket{\zeta_N}\)\quad\text{ with $p\geq 0$}.
\]
Now we will prove that each iteration $a_\lambda^{p}$ admits limit as $t\to\infty$, by induction. Clearly, $a_\lambda^{0}$ has a limit since it is constant in $t$. Let us assume that $a_\lambda^{p}(t)$ has a limit $\underline{a^{p}}_\lambda$ when $t\to\infty$ for some $p\geq 1$. By implementing the change of variable $r=t-s$ we write:
\begin{equation}
\label{3}
\begin{aligned}
a_{\lambda,n}^{p+1}(t)&= \ket{\zeta_n}+\iu \lambda\sum_{1\leq j,k\leq N}\nu_{jk}\int_0^\infty \di r\, \chi_{[0,t]}(r)\cdot\,\\
&\qquad \cdot\scal{a_{\lambda,k}^{p}(t-r)}{\ex^{-\iu (t-r) H} \,\rho_i\, \ex^{\iu (t-r) H} a_{\lambda,k}^{p}(t-r)}\cdot\\
&\qquad\cdot \scal{\zeta_j}{\ex^{\iu r H}\zeta_n} \ex^{-\iu r H} a_{\lambda,j}^{p}(t-r).
\end{aligned}
\end{equation}

Notice that in \eqref{3} the modulus of the integrand is dominated by $\abs{\scal{\zeta_j}{\ex^{\iu r H}\zeta_n}}$, which belongs to $L^1(\R)$ due to Assumption \ref{as:main}. By the dominated convergence theorem, to compute the limit of $a_\lambda^{p+1}(t)$ as $t\to \infty$, it suffices to calculate the pointwise limit of the integrand in \eqref{3}. 
Reasoning like in \eqref{dc4} we obtain
\begin{align*}
{\(\underline{a^{p+1}}\)}_{\lambda,n}&=\lim_{t\to\infty}a^{p+1}_n(t)\\ 
&= \ket{\zeta_n}+\iu \lambda\sum_{1\leq j,k\leq N}\nu_{jk}\, \scal{\underline{a^{p}}_{\lambda,k}}{W_-(H,H_{L}) \rho_i  W_-(H_{L},H) \underline{a^{p}}_{\lambda,k}}\cdot \\
&\qquad\qquad\qquad \qquad \qquad \cdot\int_0^\infty \di r \,\scal{\zeta_j}{\ex^{\iu r H}\zeta_n} \ex^{-\iu r H} \underline{a^{p}}_{\lambda,j}.
\end{align*}
Thus, for every $p\geq 0$ the limit $\lim_{t\to\infty}a_\lambda^{p}(t)$ exists and is denoted by $\underline{a^{p}}_\lambda$.

\noindent
At this point it is convenient to recognize that $\underline{a^{p+1}}_\lambda$ is the iterative computation of the fixed point of the map $\Psi\colon B_2(0)\to  B_2(0)$ defined for every  $1\leq n\leq N$ as:
\begin{equation}
\label{gc11}
\begin{aligned}
 {\(\Psi(v)\)}_n&:=\ket{\zeta_n}+\iu \lambda\sum_{1\leq j,k\leq N}\nu_{jk}\scal{v_k}{W_-(H,H_{L}) \rho_i  W_-(H_{L},H) v_k}\cdot\\
&\qquad \cdot\int_0^\infty  \di r \scal{\zeta_j}{\ex^{\iu r H}\zeta_n} \ex^{-\iu r H} v_j. 
\end{aligned}
\end{equation}
The map $\Psi$ leaves  $B_2(0)$ invariant and becomes a contraction by choosing $\lambda<\lambda_0$ as specified in \eqref{gm3}. Thus, $\lim_{p\to\infty}\underline{a^{p}}_\lambda$ exists and equals the unique fixed point $\underline{a}_\lambda$ of the map $\Psi$ in $B_2(0)$.
Then, finally we get that
\begin{equation} 
\label{6}
\lim_{t\to\infty}a_\lambda(t)=\lim_{t\to\infty}\lim_{p\to\infty}a_\lambda^p(t)=\lim_{p\to\infty}\underline{a^{p}}_\lambda=\underline{a}_\lambda,
\end{equation}
where the exchange of limits is possible thanks to the fact that the limit as $p\to\infty$ is performed with respect to norm \eqref{4}. Thus we have the identity 
\begin{equation}
\label{gc11'}
\begin{aligned}
 \underline{a}_{\lambda,n}&=\ket{\zeta_n}+\iu \lambda\sum_{1\leq j,k\leq N}\nu_{jk}\scal{\underline{a}_{\lambda,k}}{W_-(H,H_{L}) \rho_i  W_-(H_{L},H) \underline{a}_{\lambda,k}}\cdot\\
&\qquad \cdot\int_0^\infty  \di r \scal{\zeta_j}{\ex^{\iu r H}\zeta_n} \ex^{-\iu r H} \underline{a}_{\lambda,j}. 
\end{aligned}
\end{equation}

Finally, for any compactly supported function $\psi_c$ in $\Hi$, by using \eqref{6} and again the dominated convergence theorem together with Assumption \ref{as:main}, we have that
\begin{equation}\label{7}
\begin{aligned}
A_{\lambda,\infty}\psi_c&=\lim_{t\to\infty}A_\lambda(t)\psi_c\\ &= \psi_c+ \iu\lambda\sum_{1\leq j,k\leq N} \nu_{jk} \scal{\underline{a}_{\lambda,k}}{W_-(H,H_L)\,\rho_i\, W_-(H_{L},H)\underline{a}_{\lambda,k}}\cdot\\
&\qquad \qquad \qquad \qquad \cdot\int_0^\infty \di r\,\scal{\zeta_j}{\ex^{\iu r H} \psi_c}  
 \ex^{-\iu r H}\underline{a}_{\lambda,j}.
\end{aligned}
\end{equation}
\end{proof}

At this moment we can get back to \eqref{dc3}, and use \eqref{dc4} and Proposition \ref{prop:W} in order to conclude that 
\begin{align}\label{dc5}
\omega_\lambda\big (\ket{f}\bra{g}\big )
    &=\bscal{A_{\lambda,\infty}\, g}{\Big (W_-(H,H_{L})\,\rho_i\, W_-(H_{L},H)\Big )A_{\lambda,\infty}\, f}.
\end{align}
Both in \eqref{7} and \eqref{dc5}, the initial density operator $\rho_i$ only  appears sandwiched between wave operators, which project the component $\rho_s$ out. Thus the steady state only depends on the initial datum on the leads. This concludes the proof of Theorem \ref{thm:main}{(a)},{(b)}, where 
\begin{equation}\label{dc20}
    \rho_{\lambda,\infty}:=A_{\lambda,\infty}^*\, W_-(H,H_{L})\,\rho_i\, W_-(H_{L},H)\, A_{\lambda,\infty}.
\end{equation}

\subsection{Proof of Theorem \ref{thm:main}{(c)}}   We start with a lemma: 
\begin{lemma}\label{lemmahc4}
    Under the conditions of Proposition \ref{prop:W}, we have: 
    \begin{equation}\label{dc24}
    H \, A_{\lambda,\infty} =A_{\lambda,\infty}\, \big (H+V_\lambda \{\rho_{\lambda,\infty}\}\big ),\quad A_{\lambda,\infty}^* \, H   = \big (H+V_\lambda \{\rho_{\lambda,\infty}\}\big )\, A_{\lambda,\infty}^*. 
     \end{equation}
\end{lemma}
\begin{proof}
     It is enough to prove the first identity in \eqref{dc24} because the other one follows by taking the adjoint of the first one.  From Proposition \ref{prop:W} we have that $A_\lambda(t)=\ex^{-\iu tH}U^*(t)$ converges strongly to $A_{\lambda,\infty}$. Let $\delta>0$. Then 
     \begin{equation}\label{dc23}
     A_\lambda(t+\delta)=\ex^{-\iu \delta H}A_{\lambda}(t)\,  U(t)\, U^*(t+\delta).
     \end{equation}
Using \eqref{eqn:U*} \gm{and recalling that $\rho(s)=U(s)\rho_i U(s)^*$ we have:
\begin{align*}
  U(t)\, U^*(t+\delta)&=\Id +  U(t)\, \big (U^*(t+\delta)-U^*(t)\big )\\
  &=\Id +\iu \,  U(t)\int_t^{t+\delta} \di s \, U^*(s)\, \big (H+V_\lambda\{\rho(s)\}\big ) \\
  &=\Id +\iu \delta \big (H+V_\lambda\{\rho(t)\}\big ) + R_\delta(t),
\end{align*}
where we have introduced the operator
\begin{equation}
\label{eqn:defnR}    
R_\delta(t):=\iu U(t)\int_t^{t+\delta}\di s\int_t^{s}\di s_1\,\frac{\di}{\di s_1}\left[U^*(s_1) \( H+V_\lambda\{\rho(s_1)\} \)\right].
\end{equation}
By plugging the above identity in \eqref{dc23} and taking the limit $t\to\infty$ we obtain that for every $\psi\in\Hi$:
\begin{equation}
\label{eqn:R(t)}
A_{\lambda,\infty}\psi=\ex^{-\iu \delta H}A_{\lambda,\infty}\psi +\iu \delta\ex^{-\iu \delta H}A_{\lambda,\infty}\(H+V_\lambda\{\rho_{\lambda,\infty}\}\)\psi+\ex^{-\iu \delta H}\lim_{t\to\infty} A_{\lambda}(t)R_\delta(t)\psi,
\end{equation}
where the limit of the vector $A_{\lambda}(t) R_\delta(t)\psi$ is given from the existence of all the other limits.
Expanding the right-hand side of \eqref{eqn:R(t)} in $\delta$ and noticing that $A_{\lambda}(t) R_\delta(t)\psi$ is of order $\delta^2$ uniformly in $t$, we have that
\begin{equation*}
A_{\lambda,\infty}\psi=A_{\lambda,\infty}\psi +\iu \delta\Big (  A_{\lambda,\infty}\(H+V_\lambda\{\rho_{\lambda,\infty}\}\)- HA_{\lambda,\infty}\Big )\psi +\Or(\delta^2).  
\end{equation*}
Equating the linear term in $\delta$ with zero, the conclusion follows.}
\end{proof}

Now let us finish the proof of Theorem \ref{thm:main}{(c)}. Due to the intertwining property $HW_-(H,H_{L})=W_-(H,H_{L})H_L$ and because $H_L$ and $\rho_i$ commute, we have that $H$ commutes with $W_-(H,H_{L})\,\rho_i\, W_-(H_{L},H)$. Finally, \eqref{dc24} is the last ingredient for showing that $H+V_\lambda\{\rho_{\lambda,\infty}\}$ commutes with $\rho_{\lambda,\infty}$ from \eqref{dc20}.

\subsection{Proof of Theorem \ref{thm:main}{(d)}} 
Since $A_{\lambda,\infty}$ is the strong limit of a sequence of unitary operators, it follows that $A_{\lambda,\infty}$ is an isometry and 
\begin{equation}\label{gc5}
    A_{\lambda,\infty}^* A_{\lambda,\infty}=\Id.
\end{equation}
We start with a general intertwining result which holds true even if the $\beta$'s and $\mu$'s of the two leads are different. 
\begin{lemma}\label{lemmagm}
Let $\tilde{H}:=H+V_\lambda\{\rho_{\lambda,\infty}\}$ and $f\in L^\infty(\R)$. Then $f(\tilde{H})=A_{\lambda,\infty}^*\, f(H)\, A_{\lambda,\infty}.$
\end{lemma}
\begin{proof} Because $H$ has purely absolutely continuous spectrum, it is enough to prove the equality when $f$ is continuous. 
 By left multiplying  the first identity of \eqref{dc24} with $A_{\lambda,\infty}^* $ we obtain $\tilde{H}=A_{\lambda,\infty}^*\, H \, A_{\lambda,\infty}$. Also, using both identities in \eqref{dc24} we get: 
\begin{equation}\label{gc1}
A_{\lambda,\infty}\, A_{\lambda,\infty}^*\, H=A_{\lambda,\infty}\,\tilde{H}\, A_{\lambda,\infty}^*=H\, A_{\lambda,\infty}\, A_{\lambda,\infty}^*.
\end{equation}
Let us prove that $\tilde{H}^n=A_{\lambda,\infty}^*\, H^n \, A_{\lambda,\infty}$ for all $n\geq 1$. We already know it for $n=1$, hence for any $n\geq 2$ we have by induction:
\begin{align*}
\tilde{H}^n&=\big (A_{\lambda,\infty}^*\, H^{n-1} \, A_{\lambda,\infty}\big )\, A_{\lambda,\infty}^*\, H\,A_{\lambda,\infty}\, = \, A_{\lambda,\infty}^*\, H^{n-1} \, \big ( A_{\lambda,\infty}\, A_{\lambda,\infty}^*\, H\big ) \,A_{\lambda,\infty}\,
\\ &= A_{\lambda,\infty}^*\, H^{n-1} \,  \big (H\, A_{\lambda,\infty}\, A_{\lambda,\infty}^*\big ) \, A_{\lambda,\infty}\,=A_{\lambda,\infty}^*\, H^{n} \, A_{\lambda,\infty},
\end{align*}
where in the third equality we used \eqref{gc1} and in the last one \eqref{gc5}. By the Stone--Weierstrass theorem, the operator $ f\big ({H}\big )$ can be arbitrarily well approximated in the norm topology with  polynomials in ${H}$, and we are done. 
\end{proof}

If the temperatures and the chemical potentials of the two leads are equal, then the \virg{leads} component of $\rho_i$ from \eqref{eqn:rhoi} equals $f_{\rm FD}(H_L)$. The intertwining property of the operator $W_-(H,H_L)$ gives $\rho_{\lambda,\infty}=A_{\lambda,\infty}^*\, f_{\rm FD}(H)\, A_{\lambda,\infty}$. 
Then Lemma \ref{lemmagm} proves the first identity in \eqref{dc25}. 

Now let us show the second identity in \eqref{dc25}. The method is inspired by   \cite{3-tedesco} and \cite{BES}.  
Because $I_k=\iu [H,\Pi_k]=\iu [\tilde{H}\,,\, \Pi_k]=\iu [h_\tau,\Pi_k]$, in the steady state the current intensity through lead $k$ equals:
$$\omega_\lambda(I_k)={\rm Tr}\Big ( f_{\rm FD}\big (\tilde{H}\big )\, \iu [\tilde{H}\, ,\, \Pi_k]\Big ).$$
Since the operator $ f_{\rm FD}\big (\tilde{H}\big )$ can be arbitrarily well approximated in the norm topology by polynomials in $\tilde{H}$, the current intensity equals zero if we can prove that 
 $${\rm Tr}\big (\tilde{H}^n\, [\tilde{H},\Pi_k]\big )=0,\quad n\geq 0.$$
  If $n\geq 0$ then  
 $$[\tilde{H}^{n+1},\Pi_k]=\sum_{j=0}^{n}\tilde{H}^j\, [\tilde{H},\Pi_k]\, \tilde{H}^{n-j},$$
 hence by trace cyclicity 
 $${\rm Tr}\big (\tilde{H}^n\, [\tilde{H},\Pi_k]\big )=\frac{1}{n+1}{\rm Tr}\big ([\tilde{H}^{n+1},\Pi_k]\big ).$$
 Since $\tilde{H}^{n+1}=H_L^{n+1}+\text{(a finite rank operator)}$, for all $n\geq 0$, and because $H_L$ commutes with $\Pi_k$, the right-hand side of the above equality equals zero from trace cyclicity.  


\subsection{Proof of Theorem \ref{thm:main}{(e)}} 
Point {(e)} is implied by \eqref{dc20} and \eqref{hm1}.

\section{Proofs of Corollaries}
\label{sec3}
\subsection{Proof of Corollary \ref{coro2}}
 
Let 
$$\tilde{\rho}_i:={\bf 0}\, \oplus \Big (\frac{1}{\ex^{\beta_2(h_2-\mu_2)}+1}-\frac{1}{\ex^{\beta_1(h_2-\mu_1)}+1}\Big )\, \oplus \, {\bf 0}.$$ 

We see that $\tilde{\rho}_i$ also commutes with $H_L$. Denote by 
$$f_{{\rm FD},\beta_1,\mu_1}(H_L):=\frac{1}{\ex^{\beta_1(h_1-\mu_1)}+1}\, \oplus \, \frac{1}{\ex^{\beta_1(h_2-\mu_1)}+1}.$$
Since $\rho_i=\tilde{\rho}_i +f_{{\rm FD},\beta_1,\mu_1}(H_L)\,\oplus  \rho_s$ then
\begin{align*}
W_-(H,H_L)\, {\rho}_i \, W_-(H_L,H)&=W_-(H,H_L)\, \tilde{\rho}_i \, W_-(H_L,H)\, \\
&\qquad +W_-(H,H_L)\,\Big ( f_{{\rm FD},\beta_1,\mu_1}(H_L)\,\oplus  \rho_s\Big )\, W_-(H_L,H)\\
&=W_-(H,H_L)\, \tilde{\rho}_i \, W_-(H_L,H)\,+\, f_{{\rm FD},\beta_1,\mu_1}(H).
\end{align*} 
According to \eqref{dc20} and using the intertwining properties of $A_{\lambda,\infty}$ from Lemma \ref{lemmagm} we have 
\begin{align*}
\rho_{\lambda,\infty}=A_{\lambda,\infty}^*W_-(H,H_L)\, \tilde{\rho}_i \, W_-(H_L,H)A_{\lambda,\infty}+f_{{\rm FD},\beta_1,\mu_1}\big (H+V_\lambda\{\rho_{\lambda,\infty}\}\big ).
\end{align*}
Mimicking the proof of Theorem \ref{thm:main}(d), the second right-hand term from above  will not contribute to the charge current intensity, hence 
\begin{equation}\label{gc10}
\begin{aligned}
\omega_\lambda(I_1)&={\rm Tr} \big ( A_{\lambda,\infty}^*W_-(H,H_L)\, \tilde{\rho}_i \, W_-(H_L,H)A_{\lambda,\infty}\, I_1\big )\\
&={\rm Tr} \big ( \tilde{\rho}_i\, W_-(H_L,H)A_{\lambda,\infty}\, I_1\, A_{\lambda,\infty}^*\, W_-(H,H_L)\big ).
\end{aligned}
\end{equation}
 The operator $I_1=\tau\,  \iu \, \big (\ket{S_1}\bra{L_1}-\ket{L_1}\bra{S_1}\big )$ can be identified with $\tau \sigma_2$ where $\sigma_2$ is the second Pauli matrix. The operator $W_-(H_L,H)A_{\lambda,\infty}\, I_1\, A_{\lambda,\infty}^*\, W_-(H,H_L)$ is self-adjoint, has rank $2$, thus it may be written as 
\begin{equation}\label{hm4}
W_-(H_L,H)A_{\lambda,\infty}\, I_1\, A_{\lambda,\infty}^*\, W_-(H,H_L)=\tau\, \sum_{j=1}^2 (-1)^{j-1}\, \ket{f_{j,\lambda}}\bra{f_{j,\lambda}},
\end{equation}  
where 
\begin{equation}\label{gc12}
\sqrt{2}\, \ket{f_{j,\lambda}}=W_-(H_L,H)A_{\lambda,\infty}\,\ket{S_1}+\iu \, (-1)^{j} W_-(H_L,H)A_{\lambda,\infty}\,\ket{L_1}.
\end{equation}
Thus 
$$\omega_\lambda(I_1)=\tau \sum_{j=1}^2 (-1)^{j-1}\, \scal{f_{j,\lambda}}{\tilde{\rho}_i\, f_{j,\lambda}}.$$
The above scalar products can be expressed with the help of the absolutely continuous spectral measure of the lead Hamiltonian $h_2$. In fact, since $h_2$ is the Dirichlet Laplacian on the half-line, it is diagonalized by a unitary Fourier-like transform $\mathfrak{F}$, \hc{see \eqref{march1}}. Thus 
$$\omega_\lambda(I_1)= \int_{-2t_c}^{2t_c} \Big (\frac{1}{\ex^{\beta_2(E-\mu_2)}+1}-\frac{1}{\ex^{\beta_1(E-\mu_1)}+1}\Big ) \,\tau \sum_{j=1}^2 (-1)^{j-1}\,|\mathfrak{F}(\Pi_2 \, f_{j,\lambda})|^2(E)\, \di E,$$
which implies \eqref{hd1} with 
\begin{equation}\label{gc13}
2\pi\, \mathcal{T}_\lambda(E)=\tau \sum_{j=1}^2 (-1)^{j-1}\,|\mathfrak{F}(\Pi_2\, f_{j,\lambda})|^2(E).
\end{equation}
 Moreover, since $\ket{S_1}$ and $\ket{L_1}$ have compact support, we see from \eqref{7} that $A_{\lambda,\infty}\,\ket{S_1}$ and $A_{\lambda,\infty}\,\ket{L_1}$ have convergent expansions in $\lambda$ (because the fixed point $\underline{a}_\lambda$ of the contraction $\Psi$ from \eqref{gc11} has one), which via \eqref{gc12} it implies that the $\ket{f_{j,\lambda}}$'s have the same property. Thus from \eqref{gc13} we conclude that $\mathcal{T}_\lambda$ also has a convergent expansion in $\lambda$, seen as an element of $L^1([-2t_c,2t_c])$.  Under the assumptions of Lemma \ref{lemmahc1}, if $\lambda$ is small enough and  $\psi_c$ has compact support, then Appendix \ref{Ap3} implies that functions of the type $\mathfrak{F}_L W_-(H_L,H) A_{\lambda,\infty}\psi_c$ and $\mathfrak{F}_L W_-(H_L,H) \psi_c$  are actually continuous functions in the energy variable $E$. Starting from this, one can also show that the power series expansion in $\lambda$ of $\mathcal{T}_\lambda$ has coefficients which are continuous functions of $E$.

  The function $\mathcal{T}_\lambda$ depends on  $\beta$'s and $\mu$'s via $A_{\lambda,\infty}$ (see \eqref{7}). Let us investigate its continuity with respect to these parameters. We fix some state $\rho_i'$ given by  $\beta_1',\beta_2',\mu_1',\mu_2'$ where $\beta_1'$ and $\beta_2'$ may be equal to infinity. 
 
 By inspecting \eqref{gc13} and \eqref{gc12},  we observe that we only need to prove the continuity of $A_{\lambda,\infty}\psi_c$, where $\psi_c$ is either $\ket{S_1}$ or $\ket{L_1}$.  Going back to \eqref{7}, we see that this is implied by two things. The first one is the continuity in $\beta$'s and $\mu$'s of $W_-(H,H_L)\,\rho_i\, W_-(H_{L},H)$ in the weak topology. The second one is the  continuity  of the fixed point $\underline{a}_\lambda$ of the map $\Psi$ defined in \eqref{gc11}.  

  Let us prove this continuity. Denote by $\Psi'$ the map in \eqref{gc11} where $\rho_i$ is replaced by $\rho_i'$. Both $\Psi$ and $\Psi'$ are contractions, with a contraction constant $\alpha\leq \lambda/\lambda_0<1$ with $\lambda_0$ defined in \eqref{gm3}, and which is  independent of $\rho_i$. Let $\underline{a}_\lambda$ be the fixed point of $\Psi$, and let $\underline{a}'_\lambda$ be the fixed point of $\Psi'$. For every $1\leq k\leq N$ we have 
 $$\lim_{\beta\to\beta',\mu\to\mu'} \scal{\underline{a}_{\lambda,k}'}{W_-(H,H_L)\,\rho_i\, W_-(H_{L},H)\underline{a}_{\lambda,k}'}=\scal{\underline{a}_{\lambda,k}'}{W_-(H,H_L)\,\rho_i'\, W_-(H_{L},H)\underline{a}_{\lambda,k}'}, $$
because $W_-(H,H_L)\rho_iW_-(H_L,H)$ converges to $W_-(H,H_L)\rho_i' W_-(H_L,H)$ in the strong operator topology due to the absolute continuity of the spectrum of $H_L$.

 This implies that 
$\Psi(\underline{a}'_\lambda)=\Psi'(\underline{a}'_\lambda)+o(1)=\underline{a}'_\lambda +o(1)$, which means that $\underline{a}'_\lambda$ is an \virg{almost} fixed point for $\Psi$. Then: 
\begin{align*}
\Vert \underline{a}_\lambda-\underline{a}'_\lambda\Vert =\Vert \Psi(\underline{a}_\lambda)-\Psi(\underline{a}'_\lambda)+o(1)\Vert \leq \alpha\, \Vert \underline{a}_\lambda-\underline{a}'_\lambda\Vert +o(1),
\end{align*}
hence $\Vert \underline{a}_\lambda-\underline{a}'_\lambda \Vert =o(1)$, because $0\leq \alpha<1$ is uniform in $\beta$'s and $\mu$'s.

\subsection{Proof of Corollary \ref{coroHC}}

{ 
From the proof of Theorem \ref{thm:main}, by employing \eqref{dc3} and Proposition \ref{prop:W}, we know that  
$$\omega_\lambda(O_c)=\lim_{t\to\infty}\bscal{A_{\lambda,\infty}\, g}{\Big (\ex^{-\iu tH}   \,\rho_i\, \ex^{\iu tH}\Big )A_{\lambda,\infty}\, f}.$$
Denote by $\Omega_\lambda(t):=\ex^{-\iu t H_{{\rm eff},\lambda}}\ex^{\iu t H}$. Since $V_{{\rm eff},\lambda}$ lives in the space of the sample, by using Assumption \ref{as:main} 
  one can show that $\Omega_\lambda(t)$ converges strongly to $W_-(H_{{\rm eff},\lambda},H)$. Then we may write: 
$$
\omega_\lambda(O_c)=\lim_{t\to\infty}\bscal{\Omega_\lambda(t)A_{\lambda,\infty}\, g}{\Big (\ex^{-\iu t H_{{\rm eff},\lambda}}\,\rho_i\, \ex^{\iu t H_{{\rm eff},\lambda}}\Big )\Omega_\lambda(t)A_{\lambda,\infty}\, f}.$$
Reasoning like in \eqref{dc4}, we obtain that $\ex^{-\iu t H_{{\rm eff},\lambda}}\,\rho_i\, \ex^{\iu t H_{{\rm eff},\lambda}}$ converges strongly to $\rho_{{\rm eff},\lambda}$ defined in \eqref{dc15}. 
Thus 
\begin{equation}\label{hm11}
\omega_\lambda(O_c)=\bscal{W_-(H_{{\rm eff},\lambda},H)\,A_{\lambda,\infty}\, g}{\rho_{{\rm eff},\lambda}\, W_-(H_{{\rm eff},\lambda},H)\,A_{\lambda,\infty}\, f}.
\end{equation}

The final step is contained in the following lemma.
\begin{lemma}\label{lemmahc2}
  For every compactly supported function $\psi_c$ we have 
  $$\lim_{t\to \infty}\Omega_\lambda(t)A_{\lambda,\infty}\,\psi_c=W_-(H_{{\rm eff},\lambda},H)\,A_{\lambda,\infty}\,\psi_c=\psi_c+\mathcal{O}(\lambda^2).$$
\end{lemma}
\begin{proof}
First we need to compute $A_{\lambda,\infty}\,\psi_c$ up to the first order in $\lambda$. For this we need to replace $\underline{a}_{k,\lambda}$'s from \eqref{7} by their zero-th order iterates, which are nothing but the elements of the standard basis $\zeta_k$ in the sample subspace. Thus \eqref{7} and \eqref{dc1} imply: 
\begin{equation}\label{dc9}
  A_{\lambda,\infty}\,\psi_c=\psi_c +\iu \int_0^\infty \, \di r\, \ex^{-\iu r H}\, V_{{\rm eff},\lambda}\, \ex^{\iu r H}\, \psi_c +\mathcal{O}(\lambda^2). 
  \end{equation}
By differentiating and integrating back, then iterating once, we have the identity
\begin{equation}
\label{dc8}
\begin{aligned}
    &\Omega_\lambda(t)=\Id -\iu \int_0^t \di s\, \Omega_\lambda(s) \, \ex^{-\iu sH} V_{{\rm eff},\lambda}\, \ex^{\iu sH} \\
    &=\Id -\iu \int_0^t \di s\,  \ex^{-\iu sH} V_{{\rm eff},\lambda}\, \ex^{\iu sH}- \int_0^t \di s\,  \int_0^s \di u\, \Omega_\lambda(u) \, \ex^{-\iu u H} \, V_{{\rm eff},\lambda}\, \ex^{\iu (u-s)H} V_{{\rm eff},\lambda}\, \ex^{\iu sH}.
\end{aligned}
\end{equation}
Using \eqref{dc9} and the fact that $\Omega_\lambda(t)$ is unitary we get 
$$
\Omega_\lambda(t)A_{\lambda,\infty}\psi_c=\Omega_\lambda(t)\( \psi_c +\iu \int_0^\infty \, \di r\, \ex^{-\iu r H}\, V_{{\rm eff},\lambda}\, \ex^{\iu r H}\, \psi_c\) +\mathcal{O}(\lambda^2).
$$
Using the second equality in \eqref{dc8} we have 
\begin{align*}
\Omega_\lambda(t)\psi_c =& \psi_c -\iu \int_0^t \di s\,  \ex^{-\iu sH} V_{{\rm eff},\lambda}\, \ex^{\iu sH}\psi_c\\
&- \int_0^t \di s\,  \int_0^s \di u\, \Omega_\lambda(u) \, \ex^{-\iu u H} \, V_{{\rm eff},\lambda}\, \ex^{\iu (u-s)H} V_{{\rm eff},\lambda}\, \ex^{\iu sH}\psi_c.
\end{align*}
By employing the first equality in \eqref{dc8}, we obtain
\begin{align*}
\iu \Omega_\lambda(t) \int_0^\infty \, \di r\, \ex^{-\iu r H}\, &V_{{\rm eff},\lambda}\, \ex^{\iu r H}\, \psi_c=\iu \int_0^\infty \, \di r\, \ex^{-\iu r H}\, V_{{\rm eff},\lambda}\, \ex^{\iu r H}\, \psi_c\\
&+ \int_0^t \di s\int_0^\infty \di r\,\Omega_\lambda(s) \, \ex^{-\iu sH} V_{{\rm eff},\lambda}\, \ex^{\iu (s-r)H} \, V_{{\rm eff},\lambda}\, \ex^{\iu r H}\, \psi_c.
\end{align*}
Then we have  
\begin{align*}
\Omega_\lambda(t)A_{\lambda,\infty}\,\psi_c&=\psi_c -\iu \int_0^t \di s\,  \ex^{-\iu sH} V_{{\rm eff},\lambda}\, \ex^{\iu sH}\psi_c +\iu \int_0^\infty \, \di r\, \ex^{-\iu r H}\, V_{{\rm eff},\lambda}\, \ex^{\iu r H}\, \psi_c\\
&\quad -\int_0^t \di s\,  \int_0^s \di u\, \Omega_\lambda(u) \, \ex^{-\iu u H} \, V_{{\rm eff},\lambda}\, \ex^{\iu (u-s)H} V_{{\rm eff},\lambda}\, \ex^{\iu sH}\, \psi_c\\
&\quad +\int_0^t \di s\, \int_0^\infty \di r\, \Omega_\lambda(s) \, \ex^{-\iu sH} V_{{\rm eff},\lambda}\, \ex^{\iu (s-r)H}V_{{\rm eff},\lambda}\, \ex^{\iu r H}\, \psi_c +\mathcal{O}(\lambda^2).
\end{align*}
The first above integral contains the factor $V_{{\rm eff},\lambda}\, e^{\iu sH}\psi_c$ which due to Assumption \ref{as:main} will be dominated by an $L^1$ function of $s$. Taking $t\to\infty$, the first integral will cancel the second one.  The integrand of the first double integral contains the factor 
$$V_{{\rm eff},\lambda}\, \ex^{\iu (u-s)H} V_{{\rm eff},\lambda}\, \ex^{\iu sH}\, \psi_c$$
which is dominated by $\lambda^2$ times products of $L^1$ functions in $u-s$ and $s$, hence by taking $t\to \infty$, this double integral will behave like $\lambda^2$. The second double integral can be bounded in a similar way. 
\end{proof}

\subsection{Proof of Corollary \ref{Rhc1}}
When we compute the charge intensity current where the observable $O$ equals $I_1$, we may use \eqref{gc10}. Define 
\begin{equation}\label{hm10}
    \tilde{\rho}_{{\rm eff},\lambda}:=W_-(H_{{\rm eff},\lambda},H_{L})\, \tilde{\rho_i}\, W_-(H_{L},H_{{\rm eff},\lambda}).
    \end{equation}
Using the composition rule of wave operators in \eqref{gc10} we may write
\begin{align*}
\omega_\lambda(I_1)&=
{\rm Tr} \big ( \tilde{\rho}_{{\rm eff},\lambda}\, W_-(H_{{\rm eff},\lambda},H)A_{\lambda,\infty}\, I_1\, A_{\lambda,\infty}^*\, W_-(H,H_{{\rm eff},\lambda})\big ).
\end{align*}
The operator $W_-(H_{{\rm eff},\lambda},H)A_{\lambda,\infty}\, I_1\, A_{\lambda,\infty}^*\, W_-(H,H_{{\rm eff},\lambda})$ is self-adjoint, has rank $2$, and reasoning as in \eqref{hm4}, it can be written as $\tau \sum_{j=1}^2(-1)^{j-1} \ket{f_{j,{\rm eff},\lambda}}\, \bra{f_{j,{\rm eff},\lambda}}$, where 
$$\sqrt{2}\, \ket{f_{j,{\rm eff}, \lambda}}=W_-(H_{{\rm eff},\lambda},H)A_{\lambda,\infty}\,\ket{S_1}+\iu\, (-1)^{j}\, W_-(H_{{\rm eff},\lambda},H)A_{\lambda,\infty}\,\ket{L_1}.$$
    
The operator $\tilde{\rho}_{{\rm eff},\lambda}$ in \eqref{hm10} is diagonalized by a generalized Fourier transform associated with $H_{{\rm eff},\lambda}$, given by $\mathfrak{F}_{{\rm eff}, \lambda}=\mathfrak{F}_L\, W_-(H_{L}, H_{{\rm eff},\lambda}) $ where $\mathfrak{F}_L$  diagonalizes $H_L$ \hc{(see \eqref{march1} and \eqref{march0})}.   
Because $\tilde{\rho}_i$ has a non-zero component only on lead $2$, we have: 
$$\omega_\lambda(I_1)= \int_{-2t_c}^{2t_c} \Big (\frac{1}{\ex^{\beta_2(E-\mu_2)}+1}-\frac{1}{\ex^{\beta_1(E-\mu_1)}+1}\Big ) \,\tau \sum_{j=1}^2 (-1)^{j-1}\,|\Pi_2\,\mathfrak{F}_{{\rm eff}, \lambda}(f_{j,{\rm eff},\lambda})|^2(E)\, \di E.$$
Using Lemma \ref{lemmahc2} we see that 
$$\sqrt{2}\, \ket{f_{j,{\rm eff}, \lambda}}=\ket{S_1}+\iu\, (-1)^{j}\, \ket{L_1} +\mathcal{O}(\lambda^2),$$
\hc{hence the current density, seen  as an element of $L^1([-2t_c,2t_c])$, can be approximated up to an error of order $\lambda^2$ by the current density coming from  $\tilde{\rho}_{{\rm eff},\lambda}$ alone:
\begin{align*}
    2\pi \mathcal{T}_\lambda(E)&=\tau \sum_{j=1}^2 (-1)^{j-1}\,|\mathfrak{F}_{{\rm eff}, \lambda}^{(2)}(f_{j,{\rm eff},\lambda})|^2(E)\\
    &=2^{-1}\tau \sum_{j=1}^2 (-1)^{j-1}\,\left |\mathfrak{F}_{{\rm eff}, \lambda}^{(2)}(S_1+\iu\, (-1)^{j}\, L_1 )\right |^2(E)+\mathcal{O}(\lambda^2),
    \end{align*}
    where the approximation 
\begin{equation}\label{march4}
\mathcal{T}_{{\rm eff},\lambda}(E):=2^{-1}\tau \sum_{j=1}^2 (-1)^{j-1}\,\left |\mathfrak{F}_{{\rm eff}, \lambda}^{(2)}\big (S_1+\iu\, (-1)^{j}\, L_1 \big )\right |^2(E)
\end{equation}
does not contain any self-consistent terms and equals a \virg{non-interacting} Landauer-B\"uttiker transmittance where $H$ is replaced by $H_{{\rm eff},\lambda}$. }

\appendix

\section{Existence of a unique global solution for {\it U(t)}}\label{Ap1}
This Appendix is devoted to show the existence of a unique global solution for $U(t)$ solving \eqref{eqn:U}. We introduce the following norm for $\{\nu_{jk}\}_{1\leq j,k\leq N}$ entering in the definition of the Hartree potential $V_\lambda$.
Let $\Vert\nu\Vert_1$ be as in \eqref{gm3}. Then:
\begin{lemma}
\label{lem:propG}
Let $G$ be as in \eqref{eqn:G}. Then we have that
\begin{enumerate}[label=(\roman*), ref=(\roman*)]
\item \label{it:diffG}
\begin{equation}
\label{eqn:diffG}
\norm{{G}(A_1)-{G}(A_2)}\leq \norm{A_1-A_2}\(\norm{H}+\lambda\norm{\nu}_1 {\(\norm{A_1} +\norm{A_2}\)}^2\)
\end{equation}
for every operators $A_1,A_2\in\LH$.
\item \label{it:pointG}
\begin{equation}
\label{eqn:pointG}
\norm{G(A)}\leq \norm{A}\(\norm{H} +\lambda \norm{\nu}_1\norm{A}^2\)
\end{equation}
for every operator $A\in\LH$.
\end{enumerate}
\end{lemma}
\begin{proof}
First of all, we establish two preliminary inequalities.
The first one is deduced as follows. Notice that $V_\lambda\{\,\cdot\,\}$ is linear by its very definition \eqref{eqn:V} and 
\begin{equation} 
\label{eqn:boundV}
\norm{V_\lambda\{\rho\}}\leq\lambda \norm{\nu}_1 \norm{\rho}\qquad\text{for every $\rho\in\LH$}. 
\end{equation}
The second one is shown below. For any operators $A_1, A_2\in\LH$ we observe that
\begin{equation}
\label{eqn:argV}
\begin{aligned}
\norm{A_1\rho_i A_1^* -A_2\rho_i A_2^*}
&\leq \norm{A_1-A_2}\(   \norm{A_1}+\norm{A_2}\),
\end{aligned}
\end{equation}
where we have used that $\norm{\rho_i}\leq 1$ by its definition \eqref{eqn:rhoi}. Now we are ready to prove inequality \eqref{eqn:diffG}. For every operators $A_1,A_2\in\LH$, we have that
\begin{align*}
\norm{G(A_1)-G(A_2)}&\leq \norm{H}\norm{A_1-A_2}+\norm{V_\lambda \{A_1 \rho_i A_1^*\} A_1- V_\lambda \{A_2 \rho_i A_2^*\} A_2}\\
&\leq\norm{H}\norm{A_1-A_2}+\norm{V_\lambda \{A_1 \rho_i A_1^*\} A_1- V_\lambda \{A_1 \rho_i A_1^*\} A_2}\\
&\phantom{\leq}+\norm{V_\lambda \{A_1 \rho_i A_1^*\} A_2- V_\lambda \{A_2 \rho_i A_2^*\} A_2}\\
&=\norm{H}\norm{A_1-A_2}+\norm{V_\lambda \{A_1 \rho_i A_1^*\} \(A_1-A_2\)}\\
&\phantom{\leq}+\norm{V_\lambda \{A_1 \rho_i A_1^* -A_2 \rho_i A_2^*\} A_2}\\
&\leq \norm{H}\norm{A_1-A_2}+\lambda \norm{\nu}_1\norm{A_1}^2\norm{A_1-A_2}\\
&\phantom{\leq}+\lambda \norm{\nu}_1\norm{A_2} \norm{A_1-A_2}\(   \norm{A_1}+\norm{A_2}\)\\
&\leq\norm{A_1-A_2}\(  \norm{H} + \lambda \norm{\nu}_1{\( \norm{A_1}+\norm{A_2}\)}^2 \)
\end{align*}
where we have used the triangle inequality, \eqref{eqn:boundV}, the linearity of $V_\lambda\{\,\cdot\,\}$ and \eqref{eqn:argV}.
To show inequality \eqref{eqn:pointG}, we observe that
\begin{align*}
\norm{G(A)}&\leq \norm{H}\norm{ A}+\norm{ V_\lambda \{A\, \rho_i\, A^*\}}\norm{ A}
\leq\norm{ A}\( \norm{H}  +\lambda\norm{\nu}_1\norm{A}^2  \),
\end{align*}
where we have used \eqref{eqn:boundV} in the second inequality.
\end{proof}

An immediate consequence of the previous Lemma is the following result.

\begin{corollary}
\label{cor:lipG}
Let $G$ be as in \eqref{eqn:G}. Then we have that $G$ is locally Lipschitz and locally bounded. Specifically, for every $A_0\in\LH$, let $B_r(A_0)$ be the closed ball of radius $r>0$ centered at $A_0$, \ie
\[
B_r(A_0):=\{   A\in\LH: \norm{A-A_0}\leq r\}.
\]
\begin{enumerate}[label=(\roman*), ref=(\roman*)]
\item \label{it:lipG}
Let the local Lipschitz constant be defined as
\begin{equation}
\label{eqn:L}
L_{r+\norm{A_0}}:=\norm{H}+4\lambda \norm{\nu}_1{\(r+\norm{A_0}   \)}^2.
\end{equation}
Then 
\begin{equation}
\label{eqn:lipG}
\norm{G(A_1)-G(A_2)}\leq L_{r+\norm{A_0}}\norm{A_1-A_2}\qquad\text{for every $A_1,A_2\in B_r(A_0)$}.
\end{equation}
\item \label{it:maxG}
Let the local maximum be defined as
\begin{equation}
\label{eqn:M}
M_{r+\norm{A_0}}:={\( r+\norm{A_0}  \)}\( \norm{H}+\lambda\norm{\nu}_1{\(   r+\norm{A_0}\)}^2   \).
\end{equation}
Then
\begin{equation}
\label{eqn:maxG}
\norm{G(A)}\leq M_{r+\norm{A_0}}\qquad\text{for every $A\in B_r(A_0)$}.
\end{equation}
\end{enumerate}
\end{corollary}
\begin{proof}
{\it \ref{it:lipG}}
Let $A_1,A_2\in B_r(A_0)$.
In view of the triangle inequality notice that
\[
\norm{A_i}\leq\norm{A_i-A_0}+\norm{A_0}\leq r+ \norm{A_0}\qquad\text{for every $1\leq i\leq 2$}.
\]
Thus estimate \eqref{eqn:diffG} implies inequality \eqref{eqn:lipG}.

{\it \ref{it:maxG}} Let $A\in B_r(A_0)$. By using that $\norm{A}\leq r+\norm{A_0}$ and inequality \eqref{eqn:pointG}, estimate \eqref{eqn:maxG} is obtained.
\end{proof}

\begin{proposition}
\label{prop:exunU}
Let $H$ be as in \eqref{eqn:H} and $V_\lambda$ as in \eqref{eqn:V}. Then there exists a unique map $U\in C^1\([0,\infty),\LH  \)$ solving the Cauchy problem \eqref{eqn:U}, such that $U(t)$ is unitary for all $t\geq 0$.


\end{proposition}

\begin{proof}
First of all, one rewrites \eqref{eqn:U} as an integral equation
\begin{equation}
\label{eqn:intU}
U(t)=\Id -\iu \int_0^t\di s\, G(U(s)),\qquad t\geq 0
\end{equation}
where $G$ is defined in \eqref{eqn:G}. By the fact that $G$ is locally Lipschitz by Corollary \ref{cor:lipG}\ref{it:lipG}, one has that $G$ is continuous.
Thus, by fundamental theorem of calculus, finding a solution $U\in C^1\([0,\infty),\LH  \)$ of the Cauchy problem \eqref{eqn:U} is equivalent to determine a solution $U\in C\([0,\infty),\LH  \)$ of integral equation \eqref{eqn:intU}. Now, we show that there exists a unique continuous solution satisfying \eqref{eqn:intU}. 

Let $A_0$ be a linear operator with $\norm{A_0}=1$. 
By Corollary \ref{cor:lipG}, we have that the map $G$ restricted to $B_{1/2}(A_0)$ is Lipschitz with Lipschitz constant $L_{3/2}$ and is bounded by $M_{3/2}$. 
Let $0<\delta<\min\(\frac{1}{L_{3/2}},\frac{1}{2M_{3/2}}\)$. Consider the metric space
\begin{align*}
X_\delta&:=\{U\colon [0,\delta]\to \LH,\text{ $U$ continuous: }\,\sup_{t\in[0,\delta]}\norm{U(t)-\Id}\leq 1/2    \}\\
d_\infty(U,V)&:=\sup_{t\in [0,\delta] }\norm{U(t)-V(t)},\qquad \text{for every $U,V\in X_\delta $}.
\end{align*}
The space $\(X_\delta, d_\infty\)$ is complete. 
Consider the following map $\mathcal{G}_0\colon X_\delta \to C([0,\delta],\LH)$, which is defined as
\[
\mathcal{G}_0(U)(t):=\Id -\iu \int_0^t\di s\, G(U(s)),\qquad t\in[0,\delta].
\]
By using \eqref{eqn:maxG} we get that $\mathcal{G}_0$ leaves $X_\delta$ invariant. From \eqref{eqn:lipG} we obtain that 
\begin{align*}
\norm{\mathcal{G}_0(U)(t)  - \mathcal{G}_0(V)(t)}&=\norm{\int_0^t\di s\, \(G(U(s))  - G(V(s))\)}\leq \delta L_{3/2}d_\infty\( U,V \),
\end{align*}
thus $\mathcal{G}_0$ is a contraction.
Therefore, by applying Banach--Caccioppoli fixed-point theorem we conclude that there exists a unique fixed point of $\mathcal{G}_0$; namely, there exists a unique solution $U_0\in X_\delta$ of \eqref{eqn:intU} when $t\in [0,\delta]$.

Since $U_0\in X_\delta$ we have $\norm{\Id -U_0(t) }\leq 1/2$, and because $U_0(t)=\Id-\( \Id -U_0(t)\)$, it turns out that $U_0$ is invertible by the Neumann series.
Moreover, in view of  \eqref{eqn:U} and \eqref{eqn:U*} for $t\in (0,\delta)$ we have that (in operator norm topology) $\frac{\di}{\di t}\( U_0^*(t)U_0(t)  \)=0$.
Thus, by using also the continuity of $U_0^*(t)U_0(t)$, we get that
\[
U_0^*(t)U_0(t)=U_0^*(0)U_0(0)=\Id\qquad \text{for every $t\in [0,\delta]$}.
\]
Since $U_0(t)$ is invertible, we have that $U_0^*(t)\equiv {U_0(t)}^{-1}$ for any $t\in [0,\delta]$. Therefore, the operator $U_0(t)$ is unitary for all $t\in[0,\delta]. $

Analogously, we have that there exists a unique solution $U_1\in C\([\delta,2\delta],\LH  \)$ of the integral equation
\[
U_1(t)=U_0(\delta) -\iu \int_\delta^t\di s\, G(U_1(s)),\qquad t\in[\delta,2\delta].
\]
Let us sketch for completeness the argument being similar to the previous one. Setting
\[
X_{2\delta}:=\{U\colon [\delta,2\delta]\to \LH,\text{ $U$ continuous: }\,\sup_{t\in[\delta,2\delta]}\norm{U(t)-U_0(\delta)}\leq 1/2    \},
\]
we consider the corresponding map $\mathcal{G}_1\colon X_{2\delta} \to C([\delta,2\delta],\LH)$ 
\[
\mathcal{G}_1(U)(t):=U_0(\delta) -\iu \int_0^t\di s\, G(U(s)),\qquad t\in[\delta,2\delta].
\]
Observe that since $\norm{U_0(\delta)}=1$ by the unitarity of $U_0$, if $U\in X_{2\delta}$ then $U(t)\in B_{1/2}(A_0)$ for all $t\in [\delta,2\delta]$. Since the local Lipschitz constant $L_{3/2}$ and local maximum $M_{3/2}$ do depend only on $\norm{A_0}=1$ and the radius $1/2$ of the ball $B_{1/2}(A_0)$, by using $\delta$ as in the hypothesis we obtain that  $\mathcal{G}_1$ leaves invariant $X_{2\delta}$ and is a contraction on $X_{2\delta}$ itself. 
By writing 
$$
U_1(t)=\(\Id -\(U_0(\delta)-U_1(t)\)U_0^{-1}(\delta)    \)U_0(\delta),
$$
we see that the right-hand side is invertible by Neumann series.
Thus, by repeating the same argument as before for the map $U_0$, we deduce that the operator $U_1(t)$ is unitary for all $t\in [\delta,2\delta]$.
In this way we can extend the solution $U(t)$ for all $t\geq 0$.
\end{proof}


\section{Dispersive estimates for the non-interacting Hamiltonian}\label{Ap2}

In this appendix we show a dispersive estimate for the non-interacting Hamiltonian $H$ given in \ref{eqn:H}.

\begin{lemma}
\label{lem:resolventlaplacian}
Let the Dirichlet Laplacian $\Delta\su{D}$  be as in \eqref{eqn:lap}. Let $n,m\in\N$ and $\theta\in (0,\pi)$. Then one has: 
\begin{equation}
\label{eqn:resolventlaplacian}
\lim_{\eps\to 0^+}\bscal{n}{  {(\Delta\su{D}-2t_c\cos\theta -\iu \eps)}^{-1} m}=\frac{1}{2\iu\, t_c\sin(\theta)}\(\ex^{-\iu\theta(n+m+2)}-\ex^{-\iu \theta \abs{n-m}}\).
\end{equation}
If $\theta=0$ we have 
\begin{equation}\label{hc2}
\lim_{\eps\to 0^+}\bscal{n}{  {(\Delta\su{D}-2t_c -\iu \eps)}^{-1} m}=\frac{1} {2t_c}\(\abs{n-m}-(n+m+2)\),
\end{equation}
and if $\theta=\pi$ we have
\begin{equation}\label{hc3}
\lim_{\eps\to 0^+}\bscal{n}{  {(\Delta\su{D}+2t_c-\iu \eps)}^{-1} m}=\frac{(-1)^{n+m}}{2t_c}\(n+m+2-\abs{n-m}\).
\end{equation}
In particular, the map $$[-2t_c,2t_c]\ni E\mapsto G(E):=\lim_{\eps\to 0^+}\bscal{n}{  {(\Delta\su{D}-E -\iu \eps)}^{-1} m}\in \C$$
is continuous and has a continuous  extension to $\R$, being real-valued outside $\sigma(\Delta\su{D})=[-2t_c,2t_c]$. 

\end{lemma}

\begin{proof}
Let $E\in\R$ and $\eps>0$, consider $z=E+\iu \eps$. The \emph{Green function} $g_z^{\Delta\su{D}}(n,m)$ for the Dirichlet Laplacian $\Delta\su{D}$ is defined as
\[
g_z^{\Delta\su{D}}(n,m):=\bscal{n}{  {(\Delta\su{D}-z)}^{-1} m}.
\]
By the so-called \emph{method of electrostatic images}, the expression for $g_z^{\Delta\su{D}}(n,m)$ ca be expressed with the help of the Green function $g_z^{\Delta}(n,m)$ of the discrete Laplacian $\Delta$ on $\Z$ (see \eqref{eqn:greenDDelta}).
Let us recall the definition of the operator $\Delta$. For every $\psi\in \ell^2(\Z)$, 
\[
\(\Delta\psi\)(n):=t_c\(\psi(n+1)+\psi(n-1)\)\text{ for all $n\in\Z$}.
\]
By the translation invariance of the Laplacian $\Delta$, one has that $g_z^{\Delta}(n,m)\equiv g_z^{\Delta}(n-m)$. 
Now, choosing $z=2t_c\cos(\theta)+\iu\eps$ with $\theta\in (0,\pi)$ and $\eps>0$, we shall explicitly compute $g_z^{\Delta}(n-m)$. The operator $\Delta$ can be fibered by using the discrete Fourier transform $\F$, which we briefly recall in the following. For every $\psi\in\ell^2(\Z)$ being compactly supported, one defines
\[
\(\F\psi\)(k):=(2\pi)^{-1/2}\sum_{n\in\Z}\ex^{\iu k n}\psi(n),\text{ for all $k\in \R$}.
\]
The operator $\F$ extends to a unitary map from $\ell^2(\Z)\to L^2([-\pi,\pi])$ and 
\[
\big (\F \Delta\F^{-1}\big )(k)=2t_c \cos k\, .
\]
Thus
\begin{equation*}
g_z^{\Delta}(n-m)=\bscal{n}{  {(\Delta-2t_c\cos\theta -\iu \eps)}^{-1} m}=\frac{1}{2\pi}\int_{-\pi}^{\pi}\di k\,\frac{\ex^{\iu k (m-n)}}{2t_c \cos(k) -2t_c\cos(\theta) -\iu \eps }.
\end{equation*}
Notice that $g_z^{\Delta}(n-m)=g_z^{\Delta}(m-n)$ for all $n,m\in\N$ which can be seen by performing the change of variable $k'=-k$.
One can easily check that 
\begin{equation}
\label{eqn:greenDDelta}
g_z^{\Delta\su{D}}(n,m)=g_z^{\Delta}(n-m)-g_z^{\Delta}(n+m+2)\quad\text{ for all $n,m\in\N$},
\end{equation}
where the Dirichlet boundary condition on the \virg{negative half-line} is formally satisfied by putting $n=-1$ and separately $m=-1$.

Thus in the following explicit computation for $g_z^{\Delta}(n-m)$ we can suppose that $m-n\geq 0$.
Let $\mathbb{S}^1$ be the unit circle. By implementing the change of variable $z:=\ex^{\iu k}$ we obtain that
\begin{align*}
g_z^{\Delta}(n-m)&=\frac{1}{2\pi \iu}\oint_{\mathbb{S}^1}\di z\,\frac{ z^{m-n-1}  }{t_c (z+z^{-1})-2t_c\cos(\theta)-\iu \eps  }\\
&=\frac{1}{2\pi\iu t_c}\oint_{\mathbb{S}^1}\di z\,\frac{ z^{m-n}  }{ z^2- (2\cos(\theta)+\iu \eps t_c^{-1})z +1  }.  
\end{align*}
Notice that $z^2- (2\cos(\theta)+\iu \eps t_c^{-1})z +1=(z-z_+(\theta,\eps))(z-z_-(\theta,\eps))$, where
\begin{equation}
\label{eqn:defzpm}
\begin{aligned}
z_\pm \equiv z_\pm(\theta,\eps):&=\(\cos(\theta)+ \frac{\iu\eps }{2t_c}\)\pm\iu\sin(\theta)\sqrt{1-\frac{\iu\eps\cos(\theta)}{t_c\sin^2(\theta)}+\frac{\eps^2}{4t_c^2\sin^2(\theta)}}\\
&=\cos(\theta)\( 1\pm \frac{\eps}{2t_c\sin(\theta)}  \)+\iu\sin(\theta)\( \frac{\eps}{2t_c\sin(\theta)}\pm 1  \)+\Or(\eps^2),
\end{aligned}
\end{equation}
where in the last equality we have used the Taylor expansion of $\sqrt{1+x}=1+\frac{1}{2}x+\Or(x^2)$ for $x$ close to $0$. Denoting by the unit disk $D_1(0):=\{z\in\C:\quad\abs{z}<1\}$, from \eqref{eqn:defzpm} we have that if $\eps$ is small enough, then  $z_-\in D_1(0) $ and $z_+\notin D_1(0) $. 
By the residue theorem, we get that
\[
g_z^{\Delta}(n-m)
=\frac{1}{t_c}\frac{z_-^{m-n}}{z_- - z_+}
\]
and 
\[
\lim_{\eps\to 0^+}g_z^{\Delta}(n-m)
=\iu \, \frac{\ex^{-\iu \theta \abs{n-m}}}{2t_c\sin(\theta)}\quad\text{ for all $n,m\in\Z$.}
\]
Hence, for all $\theta\in (0,\pi)$ in view of \eqref{eqn:greenDDelta}, we have that for all $n,m\in\N$
\begin{align*}
g_{2t_c\cos(\theta)+\iu0^+}^{\Delta\su{D}}(n,m):=\lim_{\eps\to 0^+}g_{2t_c\cos(\theta)+\iu\eps}^{\Delta\su{D}}(n-m)=\frac{1}{2\iu t_c\sin(\theta)}\(\ex^{-\iu\theta(n+m+2)}-\ex^{-\iu \theta \abs{n-m}}\),
\end{align*}
which proves \eqref{eqn:resolventlaplacian}. The other two limits can be proved with the same residue method, and we only sketch the proof of \eqref{hc3}. Here we need to find the roots of 
$$z^2-(-2+\iu \eps t_c^{-1})\, z +1.$$ 
We have $$z_-=-1+\big (\eps/(2t_c)\big )^{1/2}\, (1-\iu) +\frac{\iu \epsi}{2t_c}+\mathcal{O}(\eps^{3/2})$$ and $$z_+=-1-\big (\eps/(2t_c)\big )^{1/2}\, (1-\iu) +\frac{\iu \epsi}{2t_c}+\mathcal{O}(\eps^{3/2}).$$ For small $\eps>0$ we have $|z_-|<1$ and $|z_+|>1$. Assuming again that $m-n\geq 0$ we have (note that $m-n$ and $m+n$ have the same parity)
\[
g_{-2t_c+\iu\, \eps}^{\Delta}(n-m)
=\frac{1}{t_c}\frac{z_-^{m-n}}{z_- - z_+}=\frac{(-1)^{m+n}\Big (1-\big (\eps/(2t_c)\big )^{1/2}(1-\iu)(m-n)+\mathcal{O}(\eps)\Big )}{2\, t_c \big (\eps/(2t_c)\big )^{1/2}(1-\iu)+\mathcal{O}(\eps^{3/2})},
\]
or 
\[
g_{-2t_c+\iu\, \eps}^{\Delta}(n-m)= \frac{(-1)^{m+n}}{2\, t_c \big (\eps/(2t_c)\big )^{1/2}(1-\iu)}
-\frac{(-1)^{m+n}(m-n)}{2t_c}+\mathcal{O}(\sqrt{\eps}).
\]
Writing a similar expansion for $n+m+2$ we see that the singular terms cancel out and we obtain the result. 

Finally, we may also compute the Green function for the case in which $|\text{Re}(z)|>2t_c$ using the same residue integration, using for example the notation $z=\pm 2t_c \cosh(\theta) +\iu \varepsilon$ with $\theta>0$. The limit $\varepsilon\downarrow 0$ exists also in this case and equals a real number. 


\end{proof}

\begin{lemma}\label{lemmahc1}
    Let us assume that the continuous matrix family 
    \[  S(E):=h_s -E - \tau^2 \sum_{j=1}^2  \bra{L_j}  (h_j-E-\iu 0_+)^{-1}   \ket{L_j}\, \ket{S_j}\bra{S_j} \]
    consists of invertible matrices in $\C^N$ for all $E\in \R$. Then the spectrum of $H$ is absolutely continuous and equals $[-2t_c,2t_c]$. Moreover,    $S(E)^{-1}$ is smooth on $(-2t_c,2t_c)$ and has convergent expansions of the type 
    \begin{equation}\label{hc1}
      S(E)^{-1}=\sum_{n\geq 0} C_n^-\, (E +2t_c)^{n/2},\quad   S(E)^{-1}=\sum_{n\geq 0} C_n^+\, (2t_c-E)^{n/2},
    \end{equation}
    for $2t_c+E>0$ and respectively $2t_c-E>0$ sufficiently small.
\end{lemma}
\begin{proof}

We denote by $\Pi_s$ and $\Pi_L$ respectively the projections on the sample subsystem and  on the two-lead subsystem.
We will use the Feshbach formula with respect to the decomposition $\Hi=\ran\Pi_s\oplus\ran\Pi_L$ to compute ${\( H-E-\iu\eps \)}^{-1}$. We introduce the reduced resolvent of both leads $R_L(E+\iu\eps):={\(\Pi_L   \(H-E-\iu\eps\)\Pi_L\) }^{-1}$, which has to be understood as the direct sum of the inverses of the operators $h_j-E-\iu\eps$ in their individual lead space. 

Let us denote by
\[  S(E+\iu \eps):=h_s -E -\iu \eps - \tau^2 \sum_{j=1}^2  \bra{L_j}  (h_j-E-\iu \eps)^{-1}   \ket{L_j}\, \ket{S_j}\bra{S_j},\quad \eps>0. \]

By the Feshbach formula, we obtain that 
\begin{equation}
\label{eqn:resformula-bis}
\begin{aligned}
{\( H-E-\iu\eps \)}^{-1}=
\begin{pmatrix}
A & B\\
C & D
\end{pmatrix}
\end{aligned}
\end{equation}
where 
\begin{equation}\label{eqn:A}
\begin{split}
A&:=S(E+\iu\, \eps)^{-1} \\
B&:=-\tau  \sum_{j=1}^2 S(E+\iu\, \eps)^{-1}\ket{S_j}\bra{L_j}R_L(E+\iu\eps)\\
C&:=-\tau\sum_{j=1}^2R_L(E+\iu\eps) \ket{L_j}\bra{S_j} S(E+\iu\, \eps)^{-1}\\
D&:=R_L(E+\iu\eps)+\tau^2 \sum_{j,k=1}^2 R_L(E+\iu\eps) \ket{L_j}\bra{S_j} S(E+\iu\, \eps)^{-1} \ket{S_k}\bra{L_k} R_L(E+\iu\eps).
\end{split}
\end{equation}
Using lemma \ref{lem:resolventlaplacian} we obtain that $\lim_{\eps\downarrow 0}S(E+\iu\eps)=S(E)$ for all $E\in \R$, and consequently $\lim_{\eps\downarrow 0}S(E+\iu\eps)^{-1}=S(E)^{-1}$.  Using the same Lemma, if $f$ has compact support then the limit  
$$\lim_{\eps\downarrow 0}\bra{f} {( H-E-\iu\eps )}^{-1}\ket{f}$$
exists for all $E\in \R$. This implies that the spectrum of $H$ is absolutely continuous.

Let us prove that $S(E)^{-1}$ is smooth on $(-2 t_c,2 t_c)$. It is enough to show that $S(E)$ is smooth. This amounts to prove that $\bra{L_j}  (h_j-E-\iu 0_+)^{-1}   \ket{L_j}$ is smooth for $1\leq j\leq 2$. Since $L_j$ has compact support, from \eqref{eqn:resolventlaplacian} and setting $E=2t_c \cos(\theta)$ the conclusion follows.

Now let us prove the expansions near $\pm 2t_c$. We do this in detail only  near $2t_c$, and we assume that $\kappa:=\sqrt{1-E/(2t_c)}$ is small. We have 
$$\cos^2(\theta)=E^2/(4t_c^2)=(1-\kappa^2)^2\quad \text{and}\quad \sin(\theta)=\kappa\, \sqrt{2-\kappa^2}.$$
Let $\tilde{\theta}(\kappa):=\arcsin\big (\kappa\, \sqrt{2-\kappa^2}\big )$. It admits a real analytic extension for $\kappa$ near zero. The right-hand side of \eqref{eqn:resolventlaplacian} is meromorphic in $\theta$ and equals an analytic function near $\theta=0$, hence replacing $\theta$ with $\tilde{\theta}(\kappa)$, the matrix element equals an analytic function of $\kappa$ near $\kappa=0$. This implies that all the elements of $S(E)$ and $S(E)^{-1}$ will have the same property, and they equal a convergent power series in $\kappa$ near zero. 
\end{proof}

\begin{remark}\label{remark-april5}
   \hc{ Here is the simplest example where the hypotheses of Lemma \ref{lemmahc1} are satisfied. Let $N=1$, \ie the small sample consists of only one \virg{dot} denoted by $\ket{\zeta}$. In this case, the two leads are coupled to the same vector $\ket{S_1}=\ket{S_2}=\ket{\zeta}$, and we assume that the coupling of the small sample with the leads is realized through $\ket{L_1}=\ket{0_1}$ and $\ket{L_2}=\ket{0_2}$, see \eqref{dc6}. The Hamiltonian of the small sample is of the form $h_s=\alpha\,  \ket{\zeta}\bra{\zeta}$ where we assume that $\alpha\in (-2t_c,2t_c)$. Then
    $$S(E)= f(E) \, \ket{\zeta}\bra{\zeta},\, \text{with}\, f(E)=\alpha -E -2 \tau^2\bscal{0_1}{  {(\Delta\su{D}-E -\iu 0_+)}^{-1} 0_1}.$$
     Hence $S(E)$ is invertible if and only if $f(E)\neq 0$. When $\tau$ is small compared to $|\alpha\pm 2t_c|$, then $f$ could have zeros only when $E\in (-2t_c,2t_c)$. Using \ref{eqn:resolventlaplacian} with $E=2t_c\cos(\theta)$ we have
    $$\bscal{0_1}{  {(\Delta\su{D}-E -\iu 0_+)}^{-1} 0_1}=\frac{-\cos(\theta) +\iu \sin(\theta)}{t_c}=-\frac{E}{2t_c^2} +\iu \frac{\sqrt{4t_c^2-E^2}}{2t_c^2}.$$
    Thus 
    $$f(E)=\alpha-E +\frac{\tau^2\, E}{t_c^2}-\iu \frac{\tau^2\,\sqrt{4t_c^2-E^2} }{t_c^2}\neq 0 \,\,  \text{when}\,\,  E\in (-2t_c,2t_c), $$
    hence $f(E)$ is never zero for real $E$. We can also see that $f(E)$ and $1/f(E)$ are smooth in the variables $\sqrt{2t_c\pm E}$ near the two thresholds, as predicted by $\eqref{hc1}$. Finally, we see that $f$ has an analytic extension to the whole complex strip ${\rm Re}(z)\in (-2t_c,2t_c)$, where $f$ has exactly one simple zero, which in the lowest order of $\tau$ it is well approximated by $\alpha +\frac{\tau^2\, \alpha}{t_c^2}-\iu \frac{\tau^2\,\sqrt{4t_c^2-\alpha^2} }{t_c^2}$. This zero is a resonance for $H$.  }

    \hc{
    If $N>1$, and if $h_s$ has only simple eigenvalues in the interval $(-2t_c,2t_c)$, a similar argument can be made for each eigenvalue by using a further Feshbach reduction to the corresponding  eigenprojection. In this case, in order to turn the eigenvalues into resonances, we need that the eigenvectors of $h_s$ have a non-trivial overlap with at least one of the vectors $\ket{S_j}$. }
\end{remark}

\begin{proposition}
\label{prop:timedecay}
Let the Hamiltonian $H$ be as in \eqref{eqn:H}, and let us assume that the conditions from Lemma \ref{lemmahc1} are satisfied. Let $f,g$ be compactly supported functions in $\Hi$. Then there exists a constant $C_{f,g}$ such that
\begin{equation}
\label{eqn:timedecay}
\abs{\scal{f}{\ex^{\iu t H}g}}\leq C_{f,g}\, t^{-3/2}\quad\text{for all $t\geq 1$}.
\end{equation}
\end{proposition}

\begin{proof}
By the polarization identity it suffices to prove \eqref{eqn:timedecay} for $f=g$. By virtue of functional calculus via the resolvent formalism, one has that
\begin{equation}
\label{eqn:resformula}
\begin{aligned}
\scal{f}{\ex^{\iu t H}f}
&=\lim_{\eps\to 0}\frac{1}{2\pi\iu}\int_{\R}\di E\, \ex^{\iu t E}\bscal{f}{\({\( H-E-\iu\eps \)}^{-1}-{\( H-E+\iu\eps \)}^{-1}\)f}\\
&=\frac{1}{\pi}\lim_{\eps\to 0}\int_{-2t_c}^{2t_c}\di E\, \ex^{\iu t  E}\im \bscal{f}{(H-E-\iu\eps )^{-1}f},
\end{aligned}
\end{equation}
where we used that the spectrum of $H$ equals $[-2t_c,2t_c]$. Using the Feshbach formula and the results of Lemmas \ref{lem:resolventlaplacian} and  \ref{lemmahc1}, we have that  
$\bscal{f}{( H-E-\iu\eps )^{-1}f}$ is bounded in $\eps$ uniformly in $-2t_c\leq E\leq 2t_c$ and we can take the limit $\eps\downarrow 0$ inside the integral. Moreover, the function 
$$F(E):=\frac{1}{\pi}\lim_{\eps\downarrow 0}\im \bscal{f}{( H-E-\iu\eps )^{-1}f}$$
is smooth on $(-2t_c,2t_c)$ and has convergent expansions in $\sqrt{2t_c\mp E}$ near $\pm 2t_c$.  Moreover, because $F(E)$ is continuous on $\R$ and equals zero outside $[-2t_c,2t_c]$ because $H$ does not have spectrum there, we must have $F(\pm 2t_c)=0$. Thus these expansions must be of the form 
$$F(E)=\sum_{n\geq 1} C_n^{\pm}\, (2t_c\mp E)^{n/2}.$$
We have $\scal{f}{\ex^{\iu t H}f}=\int_{-2t_c}^{2t_c} \di E \, F(E)\, \ex^{\iu tE}.$ We construct a smooth partition of identity on $[-2t_c,2t_c]$ consisting of $\phi_1(E)+\phi_2(E)+\phi_3(E)=1$, where $0\leq \phi_j\leq 1$, $\phi_1=1$ near $-2t_c$, $\phi_3=1$ near $2t_c$, and $\phi_2\neq 1$ only on some small enough intervals near $\pm 2t_c$, where the above expansions for $F(E)$ hold. 

Because $\phi_2(E)F(E)$ is smooth with compact support, its contribution to $\scal{f}{\ex^{\iu t H}f}$ will decay faster than any power of $t$. Let us analyze the contribution from $\phi_1$. We have 
$$\int_{-2t_c}^{2t_c} \di E \, \phi_1(E)F(E)\ex^{\iu t E}=\ex^{-2\iu t_c t}\sum_{n\geq 1}C_n^{-}\int_0^\infty \di x\, \phi_1(-2t_c +x)\, x^{n/2}\ex^{\iu tx}.$$
The function $$\phi_1(-2t_c+x)\, \Big (F(-2t_c+x)-C_1^- \sqrt{x}-C_3^- x^{3/2}\Big )$$
is $C^2$ on $(0,\infty)$, equals zero at $x=0$ and has compact support. Integrating twice by parts using $e^{\iu tx}$, we can show that the contribution coming from here decays like $t^{-2}$. We only need to treat $n=1$ and $n=3$. By the change $y=tx$, the contribution from $n=1$ equals 
$$t^{-3/2}\int_0^\infty \di y\, \phi_1(-2t_c +y/t)\, y^{1/2}\ex^{\iu y}.$$
We split the integral into one over $[0,1]$ and the other over $[1,\infty)$. The first integral has an easy limit when $t\to\infty$. The second one reads as: 
$$\int_1^\infty \di y\, \phi_1(-2t_c +y/t)\, y^{1/2}\ex^{\iu y}.$$
We integrate twice by parts using $\ex^{\iu y}$. All boundary terms at $y=1$ remain bounded in $t$. Every time we differentiate $\phi_1$ we gain a factor $1/t$, while when we differentiate $\sqrt{y}$ we gain a decay of order $1/y$. Using that $\phi_1=0$ if $y/t$ is larger than some positive number, and because $y^{-3/2}$ is integrable on $[1,\infty)$, we can show that the second integral is also uniformly bounded in $t\geq 1$, hence the term coming from $n=1$ decays like $t^{-3/2}$. In a similar way, the term with $n=3$ gives a decay like $t^{-5/2}$. Thus the contribution from $\phi_1$ decays like $t^{-3/2}$, and one can prove that the contribution from $\phi_3$ has the same behavior.
\end{proof}

\hc{
For completeness, we end this Appendix by introducing the transmission coefficient between the leads. We follow \cite{CJM, N} and we keep the details to a minimum. 
The operators on the leads, $h_1$ and $h_2$, are two copies of the Dirichlet Laplacian, which can be diagonalized by a generalized Fourier transform 
$$\mathfrak{F}:\ell^2(\N)\mapsto L^2([-2t_c,2t_c]),$$
which for some $f$ with compact support is given by:
\begin{equation}\label{march1}
\begin{aligned}
    & \big (\mathfrak{F}(f)\big )(E)=\scal{\Psi^0_E}{f}_{\ell^2(\N)},\, \Psi_E^0(n)=\frac{\sin\big ( (n+1)\theta\big )}{\sqrt{\pi t_c\sin(\theta)}}, \\ &E=2t_c\cos(\theta),\, \theta\in (0,\pi),\, n\geq 0.
     \end{aligned}
\end{equation}
Also, 
\begin{equation}\label{march0}
\mathfrak{F}_L=\mathfrak{F}\oplus \mathfrak{F}    
\end{equation}
diagonalizes $H_L=h_1+h_2$. 
In order to emphasise that we have two leads, we denote by $\Psi^0_{j,E}$, $j\in \{1,2\}$, the generalized eigenfunctions of each $h_j$. We use these functions via the Lippmann--Schwinger equation \cite{Y} in order to generate generalized eigenfunctions for the coupled operator $H=h_1+h_2+h_s+h_\tau$. Under the assumptions of Lemma \ref{lemmahc1}, they are given by (as elements of an appropriate weighted $\ell^2$ space): 
\begin{equation}\label{march3}
\ket{\Psi_{j,E}^{\pm}}=\ket{\Psi^0_{j,E}}-(H-E\mp \iu 0_+)^{-1}h_\tau\, \ket{\Psi^0_{j,E}},\quad E\in (-2t_c,2t_c),
\end{equation}
where we used that $h_s\Psi^0_{j,E}=0$. The functions $\Psi_{j,E}^+$ implement the generalized Fourier transform $\mathfrak{F}_L W_-(H_L,H)$ which diagonalizes $H$. In the spectral representation of $H_L$, the unitary scattering matrix $\mathcal{S}$ between $H$ and $H_L$ has a fiber $\mathcal{S}(E)$ which is a $2\times 2$ matrix and can be written as $\mathcal{S}(E)=\Id-2\pi \iu T(E)$, where the $T$-matrix satisfies the so-called optical theorem: 
$$T(E)-T(E)^*=-2\pi \iu T(E)\, T(E)^*.$$
Formula (2.15) in \cite{N} gives the element $T_{jk}(E)$ as: 
\begin{equation*}
    T_{jk}(E)=\scal{\Psi_{j,E}^0}{h_\tau\Psi_{k,E}^+},
\end{equation*}
or using \eqref{eqn:resformula-bis} from Lemma \ref{lemmahc1} we may write: 
\begin{equation*}
    T_{12}(E)=\tau \scal{\Psi_{1,E}^0}{L_1}\, \scal{S_1}{\Psi_{2,E}^+}=-\tau^2 \scal{\Psi_{1,E}^0}{L_1}\, \scal{L_2}{\Psi_{2,E}^0}\, \scal{S_1}{S(E)^{-1} S_2}.
\end{equation*}
Finally, the transmission coefficient between the leads is: 
\begin{equation}\label{march2}
    \mathcal{T}_{0}(E)=|T_{12}(E)|^2,
\end{equation}
which due to the optical theorem and the Lippmann--Schwinger equation can be rewritten in many equivalent forms. 
}

\section{Continuity properties of the current density}\label{Ap3}
Let us consider the fixed point $\underline{a}_\lambda\in \Hi$ from \eqref{gc11'}. Let us recall that $\mathfrak{F}_L$ denotes two copies of the generalized Fourier transform $\mathfrak{F}$ which diagonalizes both $h_j$'s \hc{(see \eqref{march0})}. Denote by  $$w_{\lambda,n}(E,\sigma):=\big (\mathfrak{F}_LW_-(H_L,H) \, \underline{a}_{\lambda,n}\big )(E,\sigma),\quad E\in [-2t_c,2t_c],\quad \sigma\in \{1,2\},\quad 1\leq n\leq N.$$
As a starting point, $w_{\lambda,n}$ is just a function belonging to $L^2([-2t_c,2t_c])\otimes \C^2$, but we will show that if $\lambda$ is small enough, these  functions are actually continuous. Let us apply the unitary $\mathfrak{F}_LW_-(H_L,H)$ on both sides of \eqref{gc11'}. By exploiting the intertwining properties of the wave operators we obtain:
\begin{equation}
\label{gcm1}
\begin{aligned}
 w_{\lambda,n}(E,\sigma)&=\big (\mathfrak{F}_LW_-(H_L,H)\ket{\zeta_n}\big )(E,\sigma)\\
 &\quad +\iu \lambda\sum_{1\leq j,k\leq N}\nu_{jk}\scal{\underline{a}_{\lambda,k}}{W_-(H,H_{L}) \rho_i  W_-(H_{L},H) \underline{a}_{\lambda,k}}\cdot\\
& \qquad \cdot\int_0^\infty  \di r \scal{\zeta_j}{\ex^{\iu r H}\zeta_n} \ex^{-\iu r E} w_{\lambda,j}(E,\sigma). 
\end{aligned}
\end{equation}
 We have  
\begin{align*}
\scal{\underline{a}_{\lambda,_k}}{W_-(H,H_{L}) \rho_i  W_-(H_{L},H) \underline{a}_{\lambda,k}}=\sum_{\sigma'=1}^2 \int_{-2t_c}^{2t_c} \di E' \, |w_{\lambda,k}(E',\sigma')|^2\, \frac{1}{\ex^{\beta_{\sigma'}(E'-\mu_{\sigma'})}+1}.
\end{align*}
Also, working under the hypothesis of Lemma \ref{lemmahc1} we can compute
$$\iu \int_0^\infty  \di r \scal{\zeta_j}{\ex^{\iu r H}\zeta_n} \ex^{-\iu r E}=-\scal{\zeta_j}{(H-E+\iu 0_+)^{-1}\zeta_n}=-\scal{S(E)^{-1}\zeta_j}{\zeta_n}.$$
Thus the functions $w_{\lambda,n}$ seen as elements of $L^2([-2t_c,2t_c])\otimes\C^2$ obey: 
\begin{align}\label{gcm2}
&w_{\lambda,n}(E,\sigma) =\big (\mathfrak{F}_LW_-(H_L,H)\ket{\zeta_n}\big )(E,\sigma)\\
 &- \lambda\sum_{1\leq j,k\leq N}\nu_{jk}\scal{S(E)^{-1}\zeta_j}{\zeta_n}\, w_{\lambda,j}(E,\sigma)\, \sum_{\sigma'=1}^2 \int_{-2t_c}^{2t_c} \di E' \, |w_{\lambda,k}(E',\sigma')|^2\, \frac{1}{\ex^{\beta_{\sigma'}(E'-\mu_{\sigma'})}+1}.\nonumber 
\end{align}
The idea is to show that the above fixed point equation also holds true when the right-hand side is seen as a map which produces  continuous functions when it acts on continuous functions. First, we need to prove that the \virg{free} term $\big (\mathfrak{F}_LW_-(H_L,H)\ket{\zeta_n}\big )(E,\sigma)$ is continuous in $E$. We have 
\begin{align*}
W_-(H_L,H)\ket{\zeta_n}&=\Pi_L W_-(H_L,H)\ket{\zeta_n}=\lim_{t\to\infty}\Pi_L \ex^{-\iu tH_L}\ex^{\iu tH}\ket{\zeta_n}\\
&=\iu \tau \sum_{\sigma'=1}^2 \int_0^\infty \di t\, \Pi_L \ex^{-\iu t H_L}\ket{L_{\sigma'}}\bra{S_{\sigma'}}\ex^{\iu t H}\ket{\zeta_n}.
\end{align*}

By applying $\mathfrak{F}_L$ we obtain
\begin{align*}\big (\mathfrak{F}_LW_-(H_L,H)\ket{\zeta_n}\big )(E,\sigma)&=\iu \tau \int_0^\infty \di t \, \ex^{-\iu t E}(\mathfrak{F}L_\sigma)(E)\, \scal{S_\sigma}{\ex^{\iu t H}\zeta_n}\\
&=-\tau\,  (\mathfrak{F}L_\sigma)(E)\,\scal{S(E)^{-1}\, S_\sigma}{\zeta_n}.
\end{align*}
Since $L_\sigma$ has compact support, $\mathfrak{F}L_\sigma$ is continuous. Also, Lemma \ref{lemmahc1} guarantees that the scalar products involving $S(E)^{-1}$ are also continuous. 

By relatively standard arguments involving Banach--Caccioppoli's fixed point theorem in the space of continuous functions defined on a compact interval, one can now show that there exists some $0<\lambda_1\leq \lambda_0$ such that the fixed point equation in \eqref{gcm2} has a unique solution in the class of continuous functions for all $0\leq \lambda\leq \lambda_1$. This solution is also varying continuously with respect to the $\beta$'s and $\mu$'s. 

Finally, if we apply $\mathfrak{F}_LW_-(H_L,H)$ to both sides of \eqref{7}, we see that at fixed $E$, the integral with respect to $r$ can be performed and gives (up to a constant) $$\scal{(H-E-\iu 0_+)^{-1}\zeta_j}{ \psi_c},$$
which is continuous in $E$ (see \eqref{eqn:resformula-bis}). Hence if $\psi_c$ has compact support, all functions of the type $\big (\mathfrak{F}_LW_-(H_L,H)A_{\lambda,\infty}\, \psi_c\big )(E,\sigma)$ are continuous in $E$. Using this in \eqref{gc13} and \eqref{gc12}, we conclude that $\mathcal{T}_\lambda(E)$ is also continuous.

\bigskip \bigskip
{\footnotesize
\begin{tabular}{ll}
(H. D.~Cornean)   
		&  \textsc{Department of Mathematical Sciences, Aalborg University} \\ 
        	&   Skjernvej 4A, 9220 Aalborg, Denmark\\
        	&  {E-mail address}: \href{mailto:cornean@math.aau.dk}{\texttt{cornean@math.aau.dk}}\\[10pt]
       
         (G.~Marcelli)   
       &  \gm{\textsc{Mathematics and Physics Department, Roma Tre University}} \\ 
        	&   \gm{Largo S. L. Murialdo 1, 00146
Roma, Italy}\\
        	&  \gm{{E-mail address}: \href{mailto:giovanna.marcelli@uniroma3.it}{\texttt{giovanna.marcelli@uniroma3.it}}}\\
\end{tabular}
}

\end{document}